\documentclass{l4dc2025_arxiv}

\usepackage{color}
\usepackage{amssymb}
\usepackage{amsmath}
\usepackage{mathtools}
\usepackage{graphicx}
\usepackage{subcaption}
\usepackage{booktabs}   
\usepackage{multirow}   

\usepackage{acro}

\usepackage{cleveref} 
\crefname{figure}{Fig.}{Fig.}
\usepackage{hyperref}

\newcommand{\eqdef}{\coloneqq}

\providecommand{\eqdef}{\stackrel{\text{\tiny def}}{=}}
\DeclarePairedDelimiter{\paren}{(}{)}
\DeclarePairedDelimiter{\vect}{(}{)}
\DeclarePairedDelimiterX{\inner}[2]{\langle}{\rangle}{#1,\,#2}
\DeclarePairedDelimiter{\norm}{\|}{\|}

\DeclarePairedDelimiterXPP{\mean}[1]{\mathbb{E}}{[}{]}{}{%
   #1%
}
\newcommand{\aset}{\mathcal{S}}

\newcommand{\reals}{\mathbb{R}}
\newcommand{\horizon}{\reals_{+}}
\newcommand{\noplayers}{n}
\newcommand{\pures}{\mathcal{A}}
\newcommand{\pure}{a}
\newcommand{\nopures}{m}
\newcommand{\actions}{\mathcal{X}}
\newcommand{\action}{x}
\newcommand{\payoff}{u}

\newcommand{\paymat}{\mathbf{A}}
\newcommand{\policy}{f}
\newcommand{\controls}{\Omega}
\newcommand{\control}{\omega}

\DeclareMathOperator{\simplex}{\Delta}

\acsetup{
  use-id-as-short = true,
  single = true,
  first-style = long
}

\DeclareAcronym{AI}{%
    long = Artificial Intelligence,
}
\DeclareAcronym{SINDy}{%
    long = Sparse Identification of Nonlinear Dynamics,
}
\DeclareAcronym{SIAR}{%
    long = Side Information Assisted Regression,
}
\DeclareAcronym{MPC}{%
    long = Model Predictive Control,
}
\DeclareAcronym{PINN}{%
    long = Physics-Informed Neural Network,
}
\DeclareAcronym{RFI}{%
    long = Robust Forward Invariance,
}
\DeclareAcronym{PC}{%
    long = Positive Correlation,
}
\DeclareAcronym{SOS}{%
    long = Sum-of-Squares,
}
\newcommand{\SINDYc}{SINDYc}
\newcommand{\SIARc}{\acs*{SIAR}c}
\newcommand{\SIARMPC}{\acs*{SIAR}\nobreakdash-\acs*{MPC}}
\newcommand{\SINDYMPC}{SINDY\nobreakdash-\acs*{MPC}}
\newcommand{\PINNMPC}{\acs*{PINN}\nobreakdash-\acs*{MPC}}
\newcommand{\figureposition}[1]{\textit{(#1)}}
\newcommand{\figl}{\figureposition{left}}
\newcommand{\figc}{\figureposition{center}}
\newcommand{\figr}{\figureposition{right}}

\def \R{\mathbb{R}}

\def \X{\mathcal{X}}

\newtheorem{innercustomthm}{Theorem}
\newenvironment{customthm}[1]
{\renewcommand\theinnercustomthm{#1}\innercustomthm}
  {\endinnercustomthm}
\newcommand*\contd{C_1}
\newcommand*\contxw{\contd(\X \times \Omega)}
\newcommand*\dist{d_{\X \times \Omega, T}}
\newcommand*\xinit{x_0}
\newcommand*\wseq{\boldsymbol{\omega}}
\newcommand*\normxw[1]{\norm{#1}_{\X \times \Omega}}
\newcommand*\ds{\, \text{d} s}
\newcommand*\txw{t, \xinit, \wseq}
\newcommand*\sxw{s, \xinit, \wseq}
\newcommand*\txwadmis{\mathcal{S}}
\newcommand*\wseqadmis{\mathcal{C}}
\newcommand*\xfw{x_{f, \wseq}}
\newcommand*\xgw{x_{g, \wseq}}
\newcommand*\tx{t, \xinit}
\newcommand*\sx{s, \xinit}
\newcommand*\dotxfw{\dot{x}_{f, \wseq}}
\newcommand*\dotxgw{\dot{x}_{g, \wseq}}

\newcommand*\ls{L_{S}}
\DeclarePairedDelimiter{\abs}{\lvert}{\rvert}
\newcommand*\fa{\; \forall \, }
\newcommand*\qfa{\quad \forall \, }
\newcommand*\qqfa{\qquad \forall \, }

\newcommand\sbullet[1][.5]{\mathbin{\vcenter{\hbox{\scalebox{#1}{$\bullet$}}}}}
\newcommand*\farg{\sbullet}
\newcommand*\innerprod[2]{\left\langle #1, #2 \right\rangle}
\newcommand*\inp{\innerprod}

\title[Learning and steering game dynamics towards desirable outcomes]{Learning and steering game dynamics towards desirable outcomes}
\usepackage{times}

\coltauthor{\Name{Ilayda Canyakmaz}$^1$ \Email{ilayda\_canyakmaz@sutd.edu.sg}\\
 \Name{Iosif Sakos}$^1$ \Email{iosif\_sakos@sutd.edu.sg}\\
 \Name{Wayne Lin}$^1$ \Email{wayne\_lin@mymail.sutd.edu.sg}\\
  \Name{Antonios Varvitsiotis}$^1$ \Email{antonios@sutd.edu.sg}\\
 \Name{Georgios Piliouras}$^2$  \Email{gpil@google.com}\\
 \addr $^1$Singapore University of Technology and Design, $^2$Google DeepMind
}

\begin{document}

\maketitle

\begin{abstract}%
Game dynamics, which describe how agents' strategies evolve over time based on past interactions, can exhibit a variety of 
undesirable behaviours including
convergence to suboptimal equilibria, cycling, and chaos.~While central planners can employ incentives to mitigate such behaviors and steer game dynamics towards desirable outcomes, the effectiveness of such interventions critically relies on accurately predicting agents' responses to these incentives---a task made particularly challenging when the underlying dynamics are unknown and observations are limited. To address this challenge, this work introduces the Side Information Assisted Regression with Model Predictive Control (SIAR-MPC) framework. We extend the recently introduced SIAR method to incorporate the effect of control, enabling it to
utilize side-information constraints inherent to game-theoretic applications to model agents' responses to incentives from scarce data.~MPC then leverages this model to implement dynamic incentive adjustments.
Our experiments demonstrate the effectiveness of SIAR-MPC in guiding systems towards socially optimal equilibria, stabilizing chaotic and cycling behaviors. Notably, it achieves these results in data-scarce settings of few learning samples, where well-known system identification methods paired with MPC show less effective results.
\end{abstract}

\section{Introduction}
Game theory provides a mathematical framework for studying strategic interactions among self-interested decision-making agents, i.e., players. The Nash equilibrium (NE) is the central solution concept in game theory, describing a state where no player has an incentive to deviate \citep{ART:N50}. 
Over time, research has shifted from simply assuming that an NE exists and players will eventually play it, to understanding \textit{how} equilibrium is reached
\citep{smale_dynamics_1976,ART:PP19}. 
This shift has led to a focus on \textit{learning} in games, exploring how strategies evolve over time based on past outcomes, adopting a dynamical systems perspective \citep{BOO:FL98,BOOK:S10}. It has been shown
that game dynamics do not necessarily converge to NE but instead can display a variety of undesirable behaviors, including cycling, chaos, Poincaré recurrence, or convergence to suboptimal equilibria \citep{akin_evolutionary_1984,sato_chaos_2002, hart_uncoupled_2003,INP:MPP18,ART:MPPS23}. 
Motivated by these challenges, our primary objective in this work is to determine:
\vspace{-0.28cm}
\begin{center}
\parbox[c]{270pt}{
\centering
    \em{
        Can we steer game dynamics towards desirable outcomes?
    }
}
\end{center}
\vspace{-0.28cm}
To address this problem, we adopt the perspective of a central planner who seeks to influence player behaviour
by designing incentives.
Our goal is to achieve this with minimal effort, ensuring that the incentives are both cost-effective and efficient. More importantly, we operate in a setting with unknown game dynamics and limited observational data, reflecting real-world scenarios where information is often incomplete or uncertain. To tackle these challenges, we introduce a new computational framework called \ac{SIAR} with \ac{MPC} (\SIARMPC{}), designed to steer game dynamics by integrating cutting-edge techniques for real-time system identification and control. In the system identification step, we predict agents' reactions to incentives, which is especially challenging
for settings where observational data is limited, difficult to obtain, or costly. To address this problem, we extend the recently introduced ~\ac{SIAR} method \citep{sakos_discovering_2023},
which was developed to identify agents' learning dynamics from a short burst of a system trajectory. To compensate for the absence of data, \ac{SIAR} searches for polynomial regressors that approximate the dynamics, satisfying side-information constraints native to game theoretical applications. These constraints represent additional knowledge about agents' learning dynamics or assumptions about their behavior beyond the observed trajectory data.
To adapt \ac{SIAR} to our needs, we broaden its scope to incorporate the influence of control inputs, 
which represent the incentives designed by the central planner. This extension enables \ac{SIAR} to model the controlled dynamics, resulting in \ac{SIAR} with control (\SIARc).
%
\begin{figure}[htb]
\centering
    \includegraphics[width=0.39\textwidth]{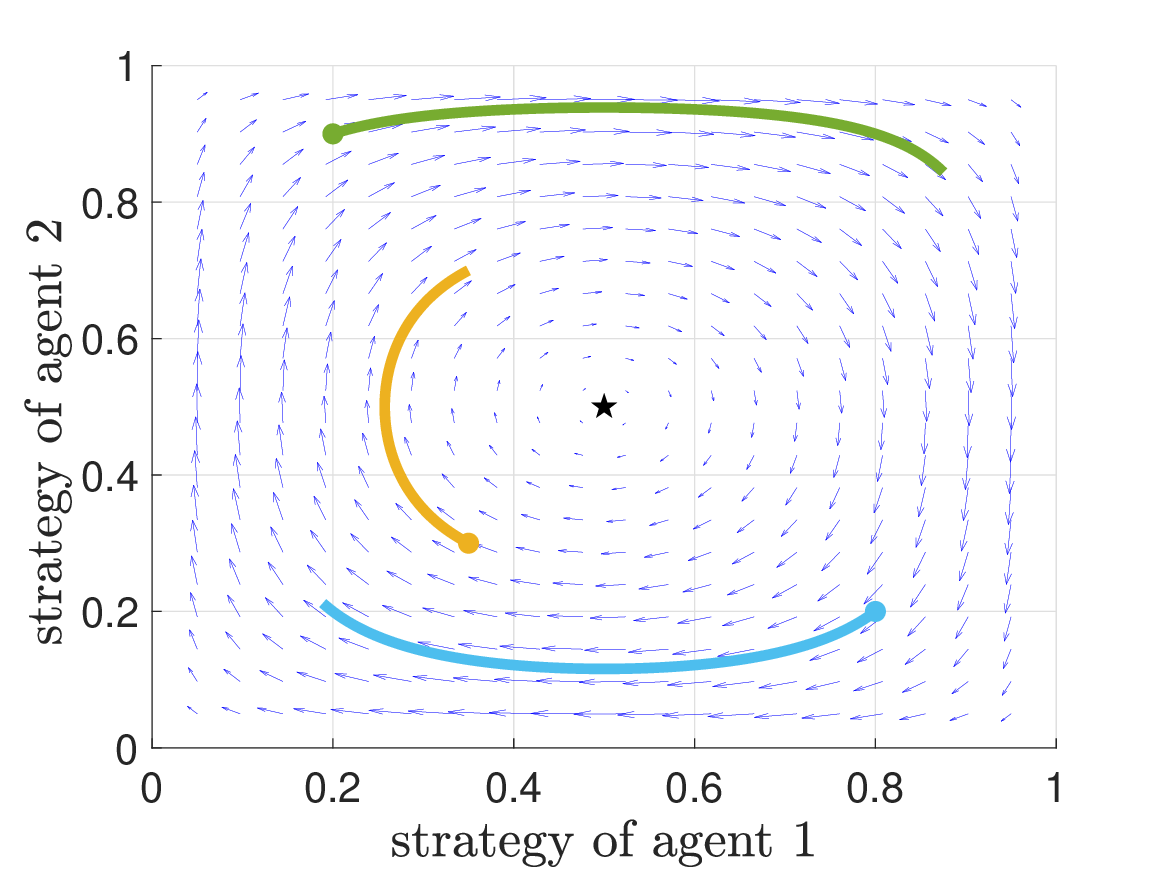}
    \includegraphics[width=0.39\textwidth]{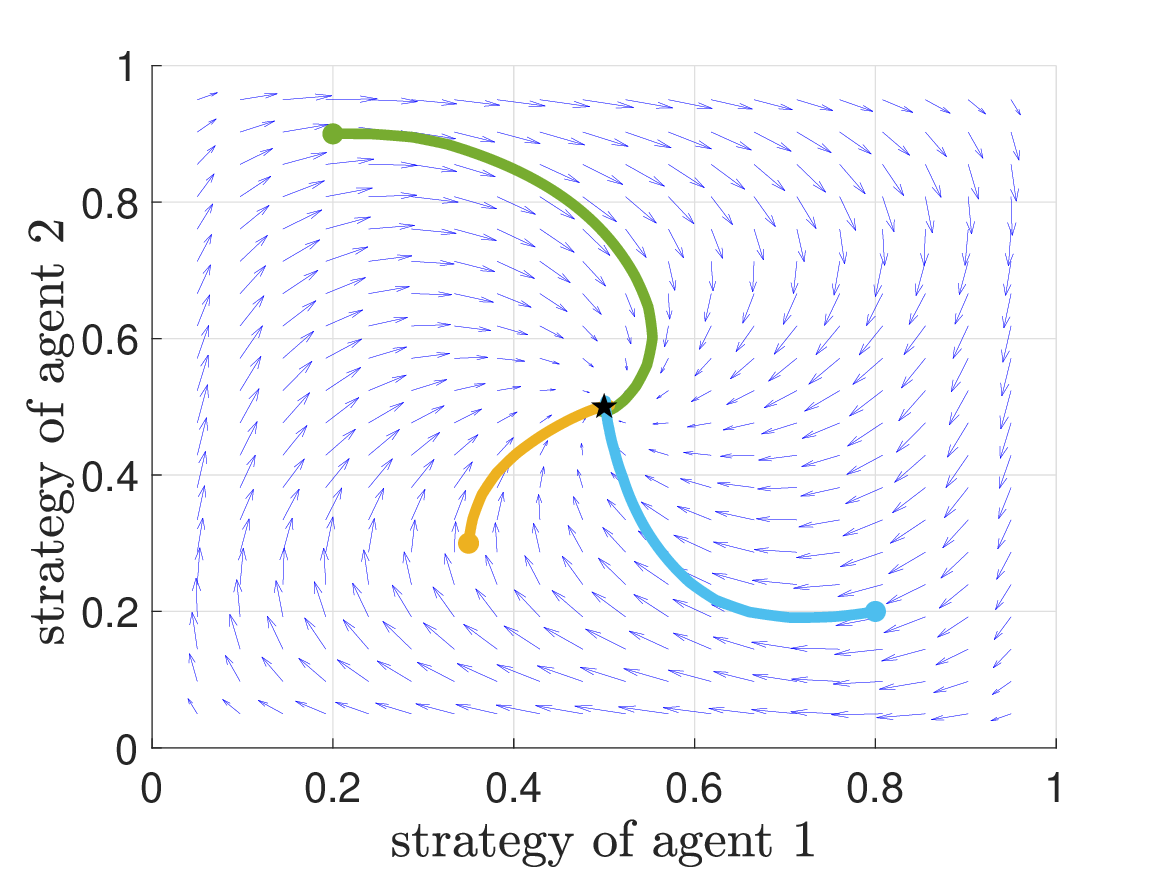}
    \caption{%
        Replicator dynamics trajectories in the matching pennies game with and without control.
        Starting from three initial conditions, \figl{} without control the system cycles around the equilibrium; \figr{} with control, trajectories are guided towards the specified equilibrium (indicated  by $\star$.)
    }
    \label{fig:MAtchingPenniesReplicatorDynamicsMPC}
\end{figure}

Once agent responses to incentives are modeled with \SIARc, we use \ac{MPC} in the subsequent control step to develop a dynamic incentive scheme that steers the system towards desirable outcomes (see \cref{fig:MAtchingPenniesReplicatorDynamicsMPC}  for an example illustrating  the impact of control). \ac{MPC} is a control technique that leverages a mathematical model of the system to predict future behavior and calculate optimal control inputs that minimize a given objective function \citep{BOO:CB07}.
A key advantage of MPC is its ability to handle input constraints, which is particularly relevant in our context since there are practical limits to the incentives that can be offered to agents. However, the effectiveness of \ac{MPC} depends on having an accurate system model, highlighting the importance of the system identification step for its successful application.~Our key contributions are:
\begin{itemize}
    \item \textbf{Framework for steering game dynamics:} We introduce the \SIARMPC{} framework for steering game dynamics towards desirable outcomes when the underlying agent behaviours are unknown and observational data are scarce.
    \item \textbf{Demonstration of performance across diverse game types:} We demonstrate the effectiveness of our framework across a diverse range of game types, from zero-sum games like Matching Pennies and Rock-Paper-Scissors to coordination games such as Stag Hunt. Our experimental results show that \SIARMPC{} can successfully steer system dynamics towards socially optimal equilibria and stabilize chaotic and cycling learning dynamics.
    \item \textbf{Demonstration of performance in a data-scarce setting:} We show that \SIARMPC{} consistently achieves convergence with low control costs in data-scarce settings characterized by limited learning samples, where other methods such as the unconstrained regression approach of \ac{SINDy} with control (\SINDYc{}) as well as \ac{PINN} (PINN) coupled with \ac{MPC} show less effective results.
\end{itemize}
\vspace{-10pt}
\section{Related Work}
\textbf{Control of Game Dynamics:} 
Recently, incentive-based control has been applied extensively in multi-player environments (see, e.g., \cite{ART:RRC18} and references therein). 
In addition, optimal control solutions have been given for specific evolutionary games and dynamics \citep{INP:PEWW18, ART:GYGC22, ART:MCL23}.
However, in the aforementioned works the player behavior is known, and to the best of our knowledge a setup where the game dynamics are not a priori known to the controller has only been explored in the recent works of \citet{misc:ZFA24} and \citet{ART:HTS24}.
\citet{misc:ZFA24} study the problem of steering no-regret agents in normal- and extensive-form games, under both full and bandit feedback.
However, in their setup the agents are adversarial and not bound to fixed dynamics, making their setup incompatible with system identification.
On the other hand, \citet{ART:HTS24} study this problem in the richer environment of Markov games under the additional assumption that the players' dynamics belong to some known finite class.
This assumption allows the controller to utilize simulators of the class of dynamics to optimize for controls that identify the exact update rule with high probability. In contrast, \acs{SIAR}-MPC does not rely on this assumption and instead takes advantage of the approximation guarantees of polynomial regression to acquire an accurate representation of an unknown system.
\vspace{4pt}
\noindent\textbf{System Identification:} The field of system identification encompasses a variety of methods, with a significant emphasis on data-driven approaches such as deep learning \citep{BOOK:BK19,INP:CSBX20}, symbolic regression \citep{ART:SL09,ART:UT20,INP:UTFN20}, and statistical learning \citep{ART:LZTM19}. A significant advancement in this field is the introduction of sparsity-promoting techniques,
most notably the \acf{SINDy}~(SINDy) method, which leverages the fact that the underlying dynamics can often be represented as a sparse combination of a set of candidate functions \citep{ART:BPK16,ART:KKB18}.
Despite their success, these data-driven methods rely on large amounts of data and often struggle to capture system dynamics accurately in data-scarce settings. Addressing this gap, \cite{ahmadi_learning_2023} explored the use of sum-of-squares (\acs{SOS}) optimization to identify dynamical systems from noisy observations of a few trajectories by incorporating contextual system information to compensate for data scarcity. 
This framework was used both in \acs{SINDy}-SI \citep{machado_sparse_2024}, an expansion of \acs{SINDy} that allows for 
polynomial nonnegativity constraints, and in the \ac{SIAR} method \citep{sakos_discovering_2023}, which learns game dynamics with few training data by leveraging
side-information constraints relevant to strategic agents. 
Physics Informed Neural Networks (PINNs) are
a
related method that
integrate knowledge of physical laws
into the training 
process
to
increase data efficiency
\citep{ART:RPK19,ART:CLS20,INP:SMC21}.

\section{Preliminaries}
In this work, we model a multi-agent system as a time-evolving normal-form game of $\noplayers$ players.
Each player~$i$ is equipped with a finite  set of strategies~$\pures_i$ of size~$\nopures_i$, and a time-varying reward function $\payoff_i : \pures \times \controls \to \reals$, where $\pures \equiv \prod_{i = 1}^{\noplayers} \pures_i$ is the game's strategy space, of the form
\vspace{-0.1cm}
\begin{equation}
    \payoff_i\paren{\pure, \control(t)} \\
        = \payoff_i(\pure, 0) + \control_{i, \pure}(t),
        \ \text{for $t \in \horizon$, $\pure \in \pures$}.
\label{eq:RewardFunction}
\end{equation}
The value $\control_{i, \pure}(t)$ denotes the control signal from the policy maker towards player~$i$ regarding the strategy profile~$\pure \eqdef \vect{\pure_1, \dots, \pure_\noplayers}$ at time~$t$, where $\pure_i \in \pures_i$.
We refer to the ensemble $\control_i(t) \eqdef \vect[\big]{\control_{i, \pure}(t)}_{\pure \in \pures}$ as the control signal of $i$ at time~$t$, and to $\control(t) \eqdef \vect[\big]{\control_i(t), \dots, \control_\noplayers(t)}$ as the system's control signal at $t$. We restrict the control signals $\control_i(t)$ in some semialgebraic sets $\controls_i$ and refer to the product $\controls \equiv \prod_{i = 1}^{\noplayers} \controls_i$ as the game's control space. The value $\payoff_i(\pure) \eqdef \payoff_i(\pure, 0)$ describes a time-independent base game, i.e., 
the reward of player~$i$ at strategy profile~$\pure$ in the absence of any control by the policy maker. The utilities 
$\payoff_i(\cdot)$, $i\!=\! 1, \dots, \noplayers$ will be considered common knowledge throughout the work. Altering the utilities is a natural method for influencing agent behavior as they encode all the information about the players' incentives in the game-theoretic model.

In addition to the above, each player~$i$ is allowed access to a set of mixed strategies $\actions_i \equiv \simplex(\pures_i)$, which is the $(\nopures_i -~1)$\nobreakdash-simplex that corresponds to the set of distributions over the pure strategies~$\pures_i$ of $i$.
The players' reward function naturally extends to the space of mixed strategy profiles $\actions \equiv \prod_{i = 1}^{\noplayers} \actions_i$ with $\payoff_i\paren[\big]{\action, \control} = \mean{\payoff_i(\pure, \control)}$ for all $\action \in \actions$ and $\control \in \controls$, where the expectation is taken with respect to the distributions $\action_1, \dots, \action_n$ over the player's pure strategies. Finally, we assume that the evolution of the above game is dictated by some controlled learning dynamics of the form
\vspace{-0.1cm}
\begin{equation}
\begin{alignedat}{2}
    \dot \action (t)
        &= \policy\paren{\action(t), \control(t)} 
            & \quad &\text{for $t \in \horizon$},\quad
    \action(0)
        &\in \actions,
\end{alignedat}
\label{eq:GameDynamics}
\end{equation}
where the update policies $\policy_i : \actions \times \controls_i \to \reals^{\nopures_i}$, $i = 1, \dots, \noplayers$ and the ensemble thereof, given by $\policy(\action, \control) \eqdef \vect[\big]{\policy_1(\action, \control_1), \dots, \policy_\noplayers(\action, \control_\noplayers)}$, are considered unknown, and are going to be discovered in the identification step of the framework described below.
Notice that the above assumption also implies that, at each time~$t$, the control signal~$\control_i(t)$ of player~$i$ is observed by that player, while the strategy profile~$\action(t)$ is observed by all the players.

Throughout this work, given a strategy profile~$\action$, we adopt the common game-theoretic shorthand $(\action_i, \action_{-i})$ to distinguish between the strategy of player~$i$ and the strategies of the other players.~Moreover, if $\action_i$ corresponds to a pure strategy $\pure_i$ of $i$, we write $(\pure_i, \action_{-i})$ to point to that fact.
\section{The SIAR-MPC Framework}
\label{sec:SIARMPC}
In this section, we describe the \SIARMPC{} framework for the real-time identification and control of game dynamics.
As outlined in the introduction, \SIARMPC{} involves two steps. 
First, the system identification step aims to approximate the controlled dynamics in (\ref{eq:GameDynamics}) using only a limited number of samples. 
Second, once the agents' reactions to payoffs are modeled,
the control step employs \ac{MPC} to steer the system towards a desirable outcome by optimizing specific objectives.
\vspace{-0.20cm}
\subsection{The System Identification Step}\label{sec:siarstep}
To model the controlled dynamics in (\ref{eq:GameDynamics}) we extend the \ac{SIAR} framework introduced in \cite{sakos_discovering_2023}, which was in turn motivated by recent results in data-scarce system identification \citep{ahmadi_learning_2023}.
\Ac{SIAR} relies on polynomial regression to approximate agents' learning dynamics of the form $\dot \action(t) = \policy\paren[\big]{\action(t)}$ based on a small number of 
observations $\action(t_k), \dot \action(t_k)$ (typically, $K = 5$ samples) taken along a short burst of a single system trajectory. 
To ensure the accuracy of the derived system model, \ac{SIAR} searches for polynomial regressors that satisfy side-information constraints native to game-theoretic applications, which serve as a regularization mechanism.

For our control-oriented scenario, we extend the \ac{SIAR} method to account for the influence of the control signal $\control(t)$. We refer to this extended method as \ac{SIAR} with control (\SIARc{}).
The aim of \SIARc{} is to model the controlled dynamics in (\ref{eq:GameDynamics}) for each agent~$i$ 
with
a polynomial vector field $p_i(\action, \control)$.
During the system identification phase,
we assemble a dataset $\action(t_k), \control(t_k), \dot \action(t_k)$, where $\action(t_k)$ represents a snapshot of the system state, $\control(t_k)$ is a randomly
generated input reflecting various possible incentives given to players (cf.~\eqref{eq:RewardFunction}), and $\dot \action(t_k)$ is the velocity at time $t_k$, which is obtained either through direct measurement (if possible) or estimated from the state variables.
The control input $\control(t_k)$ is typically taken to be normally distributed with mean-zero and low variance: the noise is used to create variety between the data samples for better system identification, while the low variance maintains a low aggregated control cost during the system identification phase.

The process of training the \SIARc{} model essentially involves solving an optimization problem to find a polynomial vector field that minimizes the mean square error relative to this dataset.
However, straightforward regression often yields suboptimal models due to the limited available samples. 
To overcome this challenge, we search over regressors that satisfy additional side-information constraints, encapsulating essential game-theoretic application features and refining the search for applicable models.
Formally, a generic \SIARc{} instance is given~by
\vspace{-0.1cm}
\begin{equation}
\begin{alignedat}{2}
    \underset{p_1, \dots, p_n}{\text{min}}& 
        &&\sum_{k = 1}^{K} \sum_{i = 1}^{n} \norm[\big]{p_i\paren[\big]{\action(t_k), \control(t_k)} - \dot \action_i(t_k)}^{2} \\
    \text{s.t.}& 
        \quad&&\text{$p_i$ are polynomial vector fields in $\action$ and $\control$} \\
        && &\text{$p_i$ satisfy side-information constraints}.
\end{alignedat}
\label{eq:SIARc}
\end{equation}
%
In this work, we utilize two specific types of side-information constraints (though a broader array is available; see, e.g., \cite{sakos_discovering_2023}).
The first side-information constraint ensures that the state space $\actions$ of the system
is robust forward invariant with respect to the controlled dynamics in (\ref{eq:GameDynamics}). 
This implies that, for any initialization $\action(0) \in \actions$, we have that $\action(t) \in \actions$ for all subsequent times $t > 0$, and for any control signal $\control(t) \in \controls$. 
To search over regressors that satisfy robust forward invariance~(\acs{RFI}), we rely on a specific characterization of the property that dictates a set remains robustly forward invariant under system (\ref{eq:GameDynamics}) only if $\policy(\action, \control)$ lies in the tangent cone of $\actions$ at $\action$ for every control $\control \in \controls$. 
Using the characterization of the tangent cone at each simplex $\actions_i$ \citep{nagumo_uber_1942}, enforcing \ac{RFI} is then reduced to verifying that, for all $\action \in \actions$ and $\control \in \controls$:
\vspace{-0.1cm}
\begin{equation}
\begin{alignedat}{2}
    \sum_{\pure_i \in \pures_i} p_{i \pure_i}(\action, \control)
        &= 0\quad \text{and}\quad p_{i \pure_i}(\action, \control)
        &\geq 0,
        & \quad &\text{whenever $\action_{i \pure_i} = 0$}.
\end{alignedat} 
\tag{\acs*{RFI}}
\label{eq:RobustForwardInvariance}
\end{equation}
The second side-information constraint is based on a fundamental assumption about agent behavior arising from their strategic nature. Agents, as strategic entities, are expected to act rationally, preferring actions that enhance their immediate benefits---a property known as positive correlation~(\acs{PC}) \citep{BOOK:S10}. 
Specifically, agents are inclined to choose actions that are likely to increase their expected utility, assuming other agents' behaviors remain unchanged, i.e., for all $\action \in \actions$ and $\control \in \controls$
\vspace{-0.1cm}
\begin{equation}
    \inner{\nabla_{\action_i} \payoff_i(\action, \control)}{p_i(\action, \control)}
        > 0,
        \; \text{whenever $p(\action, \control) \neq 0$}.
\tag{PC}
\label{eq:PositiveCorrelation}
\end{equation}
To enforce these side-information constraints computationally in our polynomial regression problem, we utilize \acs{SOS} optimization \citep{parrilo_structured_2000, prestel_positive_2001, lasserre_global_2001, parrilo_semidefinite_2003, lasserre_sum_2006}.
Both \ac{RFI} and \ac{PC} are represented as polynomial inequality or nonnegative constraints over the semialgebraic sets $\actions$ and $\controls$ (in \ac{PC}'s case, we need to relax the inequality \citep{sakos_discovering_2023}). 
Using the \ac{SOS} approach, instead of searching over polynomials $p(\action)$ that are nonnegative over a semialgebraic set $\aset \equiv \{\action\ |\ g_j(\action) \geq 0, h_\ell(\action) =~0,\, j \in [m],\, \ell \in [r]\}$, we search over polynomials $p(\action)$ that can be expressed as $p(\action) = \sigma_0(\action) + \sum_{j = 1}^{m} \sigma_j(\action) g_j(\action) + \sum_{\ell = 1}^{r} q_\ell(\action) h_\ell(\action)$ where $q_\ell$ are polynomials and $\sigma_j$ are SOS polynomials.
Such polynomials $p$ are guaranteed to be nonnegative over $\aset$, a condition that is also necessary under mild assumptions on the set $\aset$ \cite[Theorem 3.20]{laurent_sums_2008}. 
Furthermore, for any given degree $d$, we can look for \ac{SOS} certificates of degree $d$ through semidefinite programming, creating a hierarchy of semidefinite problems. Our theoretical justification for this system identification step comes from showing that the theorems of \cite{ahmadi_learning_2023} and \cite{sakos_discovering_2023}, which guarantee the usefulness of polynomial dynamics in approximation, can be readily extended to our current setting with control inputs:
\begin{theorem}[Informal]
\label{thm:SIARc_informal}
	Fix time horizon $T$, desired approximation accuracy $\epsilon > 0$, and desired side information accuracy $\delta > 0$. For any continuously differentiable dynamics $f$ satisfying side information constraints, there exists polynomial dynamics $p$ that $\delta$-satisfies the same side-information constraints and is $\epsilon$-close to $f$.
\end{theorem}
\vspace{-0.05cm}
Here, $f$ and $p$ are $\epsilon$-close if for any initial point and 
sequence of controls $\wseq\!:\![0,T]\!\to \Omega$,
the distance between the trajectories and the distance between their velocities at any time $t\!\in\![0, T]$ are both upper bounded by $\epsilon$. $\delta$-satisfiability of the side information constraints indicates approximately satisfying them to some tolerance $\delta$ and has an SOS certificate, meaning that we can perform regression over the space of polynomials that $\delta$-satisfy the side-information constraints. A formal statement of this theorem, together with its proof, is presented in Appendix \ref{app:SIARc}.
\vspace{-0.1cm}
\subsection{The Control Step}
\vspace{-0.1cm}
After estimating the controlled dynamics which describes how players' strategies ($\action$) respond to the incentives ($\control$), our next goal is to steer the system towards desirable outcomes. 
To achieve this, we employ \ac{MPC}, which formulates an optimization problem to identify the optimal sequence of control actions over a defined horizon, subject to constraints on control inputs. In our context, these control constraints represent practical limits on the incentives that can be offered to agents.
\Ac{MPC} leverages a mathematical model of the system to predict future behavior over a specified prediction horizon~$T$. Each prediction is based on the current state measurement $\action(t)$ and a sequence of future control signals $\control_t \eqdef \{\control_{0|t}, \dots, \control_{N - 1|t}\} \subset \controls$, which is calculated by solving a constrained optimization problem. 
 Here, $N$ is the number of control steps within the control horizon $T \eqdef N \cdot \Delta t$ that determines the time period over which the control sequence is optimized. 
The optimal control sequence is  obtained by solving the constrained optimization problem
\vspace{-0.1cm}
\begin{equation}
\begin{alignedat}{2}
    \underset{\control_t}{\text{min}}& 
        &&\sum_{n = 0}^{N} \norm{\action_{n|t} - \action^{*}}^{2} + \alpha \sum_{n = 0}^{N - 1} \norm{\control_{n|t}}^{2}+\beta\sum_{n = 1}^{N} \norm{\control_{n|t} -\control_{n - 1|t}}^{2} \\
    \text{s.t.}& 
        \quad&&\control_{n|t} \in \controls,\, \  0 \leq n \leq N \\
        && &\action_{n + 1|t} = \action_{n|t} + \Delta t \cdot p(\action_{n|t}, \control_{n|t}),\, \  1 \leq n \leq N\\
        && &\action_{0|t} = \action(t).
\end{alignedat} 
\end{equation}
Here, $\action_{n|t}$, $n = 1, \dots, N$ correspond to the forecasted trajectory, and $\action^{*}$ is a desirable target state.
The first term of the objective function penalizes deviations of the predicted states $\action_{n|t}$ from the target value $\action^{*}$, the second term accounts for the control effort, and the final term penalizes large variations in control signal.
The first control signal $\control_{0|t}$ is then applied to the system and the optimization is repeated at time $t + \Delta t$ once the new state measurement $\action(t + \Delta t)$ is obtained. 

\vspace{-0.2cm}
\section{Experiments}\label{sec:exp}
\vspace{-0.1cm}
In this section, we evaluate the \SIARMPC{} framework across various types of normal-form games, each of independent interest,
and compare its performance to pairing \ac{MPC} with solutions obtained from \SINDYc{}
and \acs{PINN}.
The implementation of \SINDYc{} follows the approach described in \cite{ART:KKB18}, employing sequential least squares for solving the sparse regression problem as detailed in \cite{ART:BPK16}.~For \ac{PINN}, side-information constraints are integrated into the neural network training as terms in the loss function that penalize violations of the desired constraints. 
Given the limited number of training samples---in most cases $5$---we are restricted to a simple neural network architecture with two hidden layers of size $5$. This limitation in the neural network's expressivity is counterbalanced by the simplicity of the ground-truth update policies~$\policy$ (which, in most cases, are polynomial). As activation functions we use the $\tanh$ function and the side-information constraints are enforced using $2,500$ collocation points. 
Further details on the construction of the loss function and the generation of collocation points can be found in Appendix \ref{app:pinn}.
\begin{figure*}[b!]
\centering
    \includegraphics[width=\textwidth]{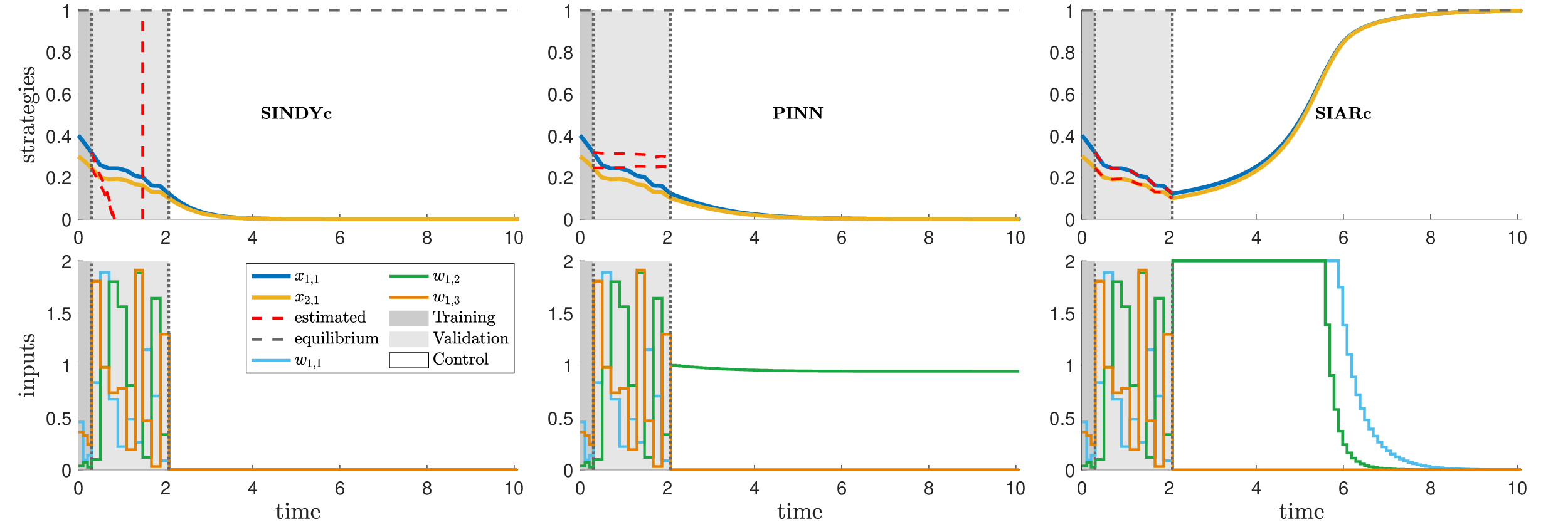}
    \caption{%
        Performance comparison of \SINDYMPC{}~\figl{}, \PINNMPC{}~\figc{}, and \SIARMPC{}~\figr{} in steering the replicator dynamics for the stag hunt game.
    }
\label{fig:ReplicatorDynamicsStagHunt}
\end{figure*}
\vspace{-0.3cm}
\subsection{Stag Hunt Game}\label{sec:exp_sh}
\vspace{-0.1cm}
The stag hunt game is a two-player, two-action coordination game that models a strategic interaction in which both players benefit by coordinating their actions towards a specific superior outcome (the hunt of a stag).
However, if that outcome is not possible, each player prefers to take advantage of the lack of coordination and come out on top of their opponent by choosing the alternative (hunt a rabbit by themselves) rather than coordinating to an inferior outcome (hunt a rabbit together).
While we perform experiments over various payoff matrices corresponding to stag hunt games in Section \ref{sec:exp_diff}, for concreteness in this section we consider the example where the players' reward functions $\payoff(\cdot) \eqdef \payoff(\cdot, 0)$ in the absence of control (see \eqref{eq:RewardFunction}) are given by
\vspace{-0.1cm}
\begin{equation}
    \payoff_1(\pure) 
        = \payoff_2(\pure)
        = \paymat_{\pure_1, \pure_2},
    \; \text{where} \; 
    \paymat = \begin{pmatrix}
        4 & 1 \\
        3 & 3
    \end{pmatrix}.
\label{eq:StagHuntGame}
\end{equation}
For expositional purposes, we restrict the game's evolution to a subset of symmetric two-player two-action games given by the control signals $\control_{1, \pure_1, \pure_2}(t) = \control_{2, \pure_2, \pure_1}(t) \in [0, 2]$ for all $t$.
Furthermore, for tractability purposes we fix $\control_{1, 2, 2}(t) \eqdef 0$.
We assume that
the agents' true behavior follow
the replicator dynamics given, for player~$i$ and action~$\pure_i$ of $i$, by the polynomial update policies
\vspace{-0.1cm}

\begin{equation}
    \policy_{i, \pure_i}\paren{\action, \control}
        = \action_{i, \pure_i} \paren{\payoff_i(\pure_i, \action_{-i}, \control) - \payoff_i(\action, \control)},\quad \text{for all}\quad \action \in \actions\quad \text{and}\quad \control \in \controls.
\label{eq:StagHuntReplicatorDynamics}
\end{equation}
We search for $\policy$ using the \SIARc{} framework.
Specifically, for each $i$  and $\pure_i$, we are going to search for polynomial $p_{i, \pure_i} : \actions \times \controls \to \reals$ such that $p_{i, \pure_i}\paren[\big]{\action(t), \control(t)} \approx \policy_{i, \pure_i}\paren[\big]{\action(t), \control(t)}$ for all $t \in \horizon$.
As side-information constraints, we are going to impose the \ac{RFI} and \ac{PC} properties as given in the previous sections.
Then, by substituting (\ref{eq:StagHuntGame}) to (\ref{eq:PositiveCorrelation}) we have the following \SIARc{} problem
\vspace{-0.2cm}
\begin{equation}
\begin{alignedat}{2}
    \underset{p}{\text{min}}& 
        &&\sum_{k = 1}^{K} \norm[\big]{p\paren[\big]{\action(t_k), \control(t_k)} - \dot \action(t_k)}^{2} \\
    \text{s.t.}& 
        \quad&&p_{i, 1}(\action, \control) + p_{i, 2}(\action, \control) = 0,\, \forall i \\
        && &p_{i, 1}\paren[\big]{\vect{0, 1}, \action_{-i}, \control} \geq 0,\, \forall i \\
        && &p_{i, 2}\paren[\big]{\vect{1, 0}, \action_{-i}, \control} \geq 0,\, \forall i \\
        && &v_{i, 1}(\action, \control) p_{i, 1}(\action, \control) + v_{i, 2}(\action, \control) p_{i, 2} \geq 0,\, \forall i,
\end{alignedat}
\label{eq:SIARStagHunt}
\end{equation}
where $\action \in \actions$, $\control \in \controls$, and $v_{i, \pure_i} : \actions \times \controls \to \reals$ are given by
\vspace{-0.1cm}
\begin{equation}
\begin{alignedat}{2}
    v_{1, 1}(\action, \control)
        &= (4 + \control_{1, 1}) \action_{2, 1} + (1 + \control_{1, 2}) \action_{2, 2}, \qquad \quad
    v_{1, 2}(\action, \control)
        &= (3 + \control_{2, 1}) \action_{2, 1} + 3 \action_{2, 2}, \\
    v_{2, 1}(\action, \control)
        &= (4 + \control_{1, 1}) \action_{1, 1} + (1 + \control_{1, 2}) \action_{1, 2}, \qquad \quad
    v_{2, 2}(\action, \control)
        &= (3 + \control_{2, 1}) \action_{1, 1} + 3 \action_{1, 2}.
\end{alignedat}
\end{equation}
Since the update policies~$\policy_{i, \pure_i}$ in (\ref{eq:StagHuntReplicatorDynamics}) correspond to the replicator dynamics, the solution to the above optimization problem can be recovered by a $7$\nobreakdash-degree \ac{SOS} relaxation \citep{sakos_discovering_2023}.
\Cref{fig:ReplicatorDynamicsStagHunt}~\figr{} shows the performance of the \SIARMPC{} using this solution as a model for the \ac{MPC} method.
In the top panel of the figure, we have the trajectory
$\action(t)$ initialized at $\action_0 = \vect{0.4, 0.3} $ corresponding to the control signal~$\control(t)$ depicted in the bottom panel.
The plot is divided into three sections corresponding to the system identification phase, an evaluation period, and a control phase.
In the first section, we set the control signals to normally distributed noise with mean zero (bounded in $\controls$) and a sample of $K = 4$ datapoints from the resulting trajectory.
At the end of this phase, we solve the optimization problem in (\ref{eq:SIARStagHunt}) and acquire a model of the system's update policies.
In the second section of the plot, we compare the ground-truth dynamics with the model's predicted trajectory (in dashed red lines); here, the control signal is chosen randomly.
Finally, in the steering phase, we steer the ground truth using the output of the \ac{MPC} as the control signal: the objective is to steer the system to the superior Nash equilibrium (NE) of the stag hunt game at $\action^{*}_{1, 1} = \action^{*}_{2, 1} = 1$, which is achieved by the \SIARMPC{} framework at $t \approx 8$.
In comparison, the \SINDYMPC{} and \PINNMPC{} solutions (\cref{fig:ReplicatorDynamicsStagHunt}~\figl{} and 
\cref{fig:ReplicatorDynamicsStagHunt}~\figc{})
fail to complete the steering objective.
%
\begin{figure*}[!]
\centering
    \includegraphics[width=0.497\textwidth]{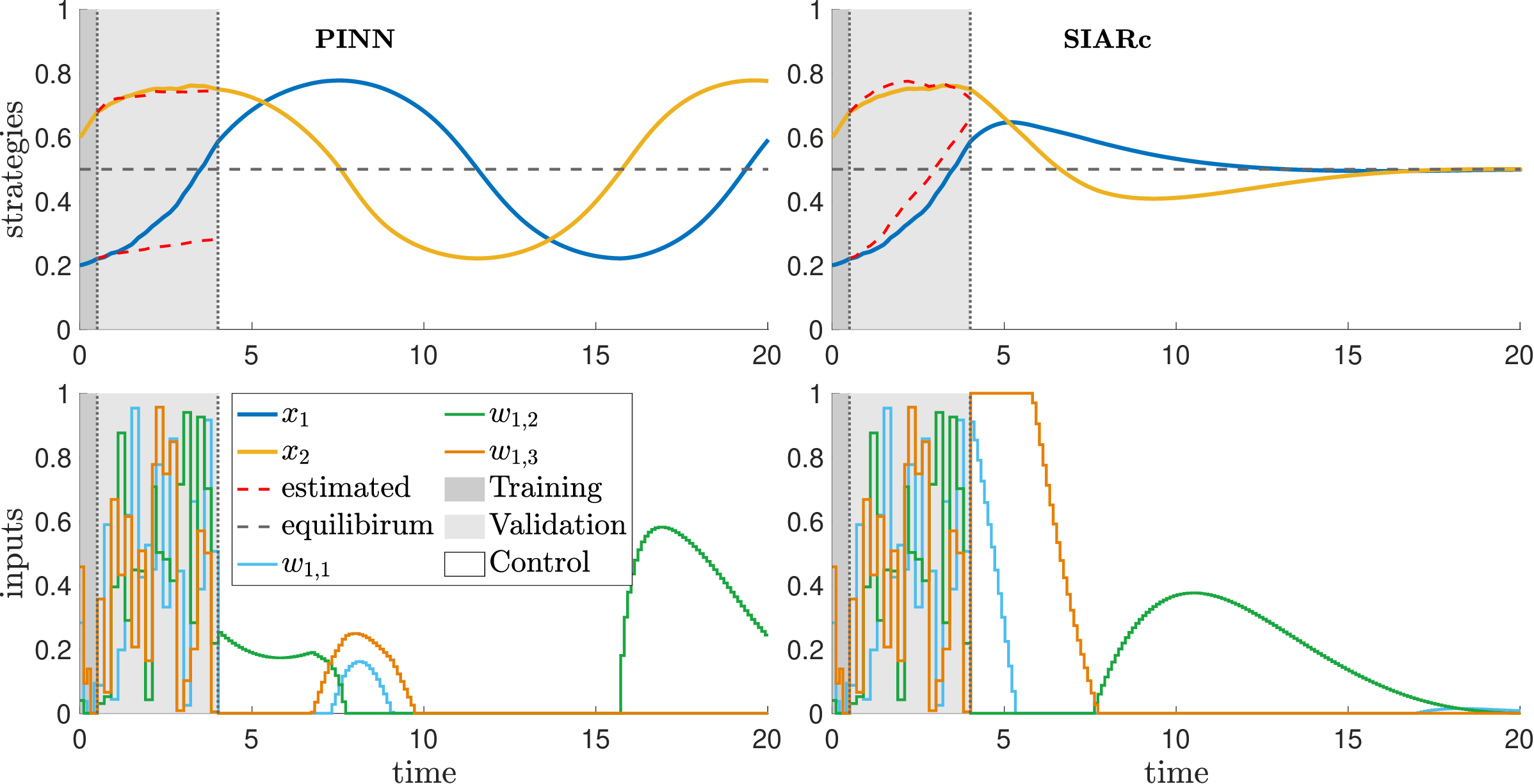}
    \includegraphics[width=0.497\textwidth]{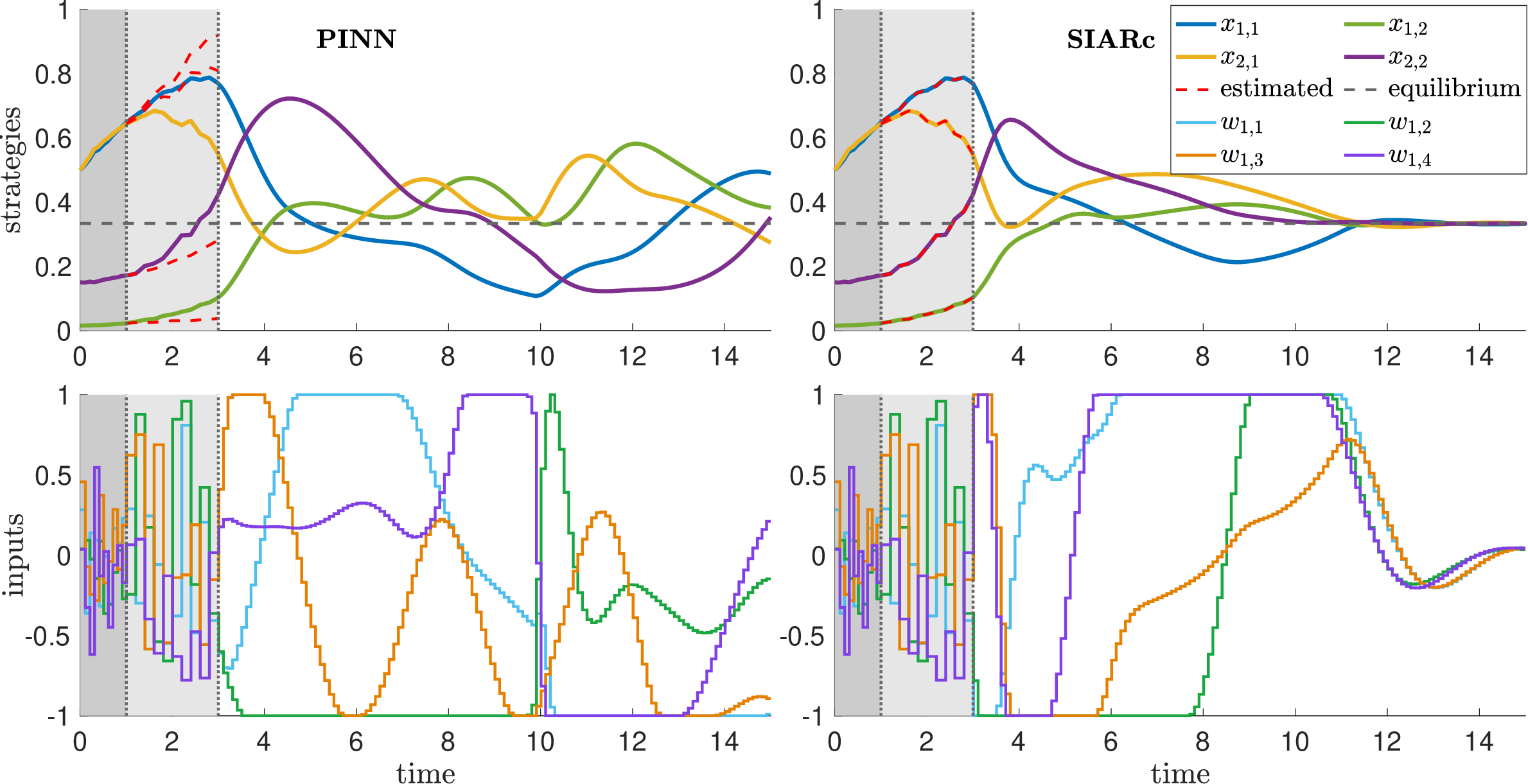}
    \caption{%
        Performance comparison of \PINNMPC{} and \SIARMPC{} steering the log-barrier dynamics for the matching pennies game~\figl{} and the replicator dynamics for a $0.25$\nobreakdash-RPS game~\figr{}.
    }
\label{fig:ZeroSumGames}
\end{figure*}
\vspace{-0.20cm}
\subsection{Zero-Sum Games \& Chaos}\label{sec:exp_zs}
The stag hunt game provides an ideal setting for steering game dynamics to a NE due to the existence of a socially optimal NE with a positive-measure basin of attraction. In contrast, the class of zero-sum games lacks this property. 
Well-known game dynamics, such as replicator and log-barrier, are known to exhibit undesirable behaviors in zero-sum games, including cycling (see \cref{fig:MAtchingPenniesReplicatorDynamicsMPC}) and chaos.
To address these challenges, we present two examples to demonstrate the performance of \SIARMPC{} in steering log-barrier dynamics
in the matching pennies game and chaotic replicator dynamics in an $\epsilon$\nobreakdash-perturbed rock-paper-scissors ($\epsilon$-RPS) game.
\paragraph{Matching Pennies Game}
The matching pennies game is a two-player, two-action zero-sum game where one player benefits by the existence of coordination among the two, while the other player benefits by the lack thereof.
Formally, a matching pennies game is encoded in the uncontrolled players' reward functions in (\ref{eq:RewardFunction}) by
\vspace{-0.3cm}
\begin{equation}
    \payoff_1(\pure) 
        = - \payoff_2(\pure)
        = \paymat_{\pure_1, \pure_2},
    \; \text{where} \;
    \paymat = \begin{pmatrix*}[r]
          1 & - 1 \\
        - 1 &   1
    \end{pmatrix*}.
\label{eq:MatchingPenniesGame}
\end{equation}
For similar reasons as in the previous example, we are going to restrict the game's evolution to a subset of two-player two-action zero-sum games given by $\control_{1, \pure}(t) = \control_{2, \pure}(t)\in [0, 1]$ for all $t$, and set $\control_{1, 2, 2}(t) \eqdef 0$.
The log-barrier dynamics of the above time-varying game are, for player~$i$ and action~$\pure_i$ of $i$, given by the rational update policies
\vspace{-0.2cm}
\begin{equation}
    \policy_{i, \pure_i}\paren[\big]{\action, \control}
        = \action_{i,\pure_i}^{2} \paren[\bigg]{\payoff_{i, \pure_i} - \frac{\action_{i, 1}^{2} \payoff_{i, 1} +\action_{i, 2}^{2} \payoff_{i, 2}}{\action_{i, 1}^{2} + \action_{i, 2}^{2}}}
\label{eq:MatchingPenniesLogBarrierDynamics}
\end{equation}
for all $t \in \horizon$ and $\action \in \actions$, where the shorthand $\payoff_{i, \action_j} \eqdef \payoff_i(\action_j, \action_{-i}, \control)$ is used for compactness.
Observe that, as is the case for the updated policies of the replicator dynamics in (\ref{eq:StagHuntReplicatorDynamics}), $\policy_i$ depends on $\control(t)$ through $\payoff_i$.
In \cref{fig:ZeroSumGames}~\figl{}, we show that \SIARMPC{} solution is able to steer a trajectory $\action(t)$ of the above system initialized at $\action(0) = \vect{0.2, 0.6}$ to the unique mixed NE $\action^{*}_{1, 1} = \action^{*}_{2, 1} = 1 / 2$ of the matching pennies game with only $K = 6$ training samples.

%
\paragraph{$\epsilon$-RPS Game}
In this example, we use the \SIARMPC{} method to steer the replicator dynamics in an $\epsilon$\nobreakdash-RPS game, a two-player, three-action zero-sum game where the replicator dynamics exhibit chaotic behavior \citep{sato_chaos_2002, hu_chaotic_2019}.
In a nutshell, this means that any two initialization of the system---even the ones that are infinitesimally close to each other---may lead to completely different trajectories.
In other words, the accurate estimation of the agents' update policies is futile due to the finite precision of any numerical method.
An $\epsilon$\nobreakdash-RPS game 
is encoded by 
\vspace{-0.15cm}
\begin{equation}
    \payoff_1(\pure) 
        = - \payoff_2(\pure)
        = \paymat_{\pure_1, \pure_2},
    \; \text{where} \;
    \paymat
        = \begin{pmatrix*}[r]
            \epsilon &      - 1 &        1 \\
                   1 & \epsilon &      - 1 \\
                 - 1 &        1 & \epsilon
        \end{pmatrix*}.
\end{equation}
%
We consider time-evolving games in the subset of two-player three-action zero-sum games given by $\control_{1, \pure}(t) = \control_{2, \pure}(t)\in [-1, 1]$ for all $t$, and only four non-zero signals, namely, $\control_{1, 1, 2}(t)$, $\control_{1, 1, 3}(t)$, $\control_{1, 2, 1}(t)$, and $\control_{1, 3, 1}(t)$. This restriction on the controllers to be nonzero only on a subset of the outcomes is to demonstrate steerability even with incomplete controls.
In \cref{fig:ZeroSumGames}~\figr{}, we show the successful steering of the chaotic replicator dynamics \eqref{eq:StagHuntReplicatorDynamics} in the $0.25$\nobreakdash-RPS game towards the unique mixed NE $\action^{*}_1 = \action^{*}_2 = \vect{1 / 3, 1 / 3, 1 / 3}$ with only $K = 11$ learning samples.

\vspace{-0.20cm}
\subsection{Experiments with different initializations and payoff matrices}
\label{sec:exp_diff}

Finally, we conduct simulations across a large 
set of settings to gain a statistically significant understanding of each methods' performance. 
Table \ref{tab:summary} presents the performance of each method across 100 initial conditions for the three examples discussed in Section \ref{sec:exp}. Table 2 shows the performance of each method under varying payoff matrices for stag hunt and zero-sum games. The results demonstrate that \SIARMPC{} consistently achieves lower error values and control cost compared to other methods. All results here have been averaged across the state variables for compactness; more details on the experiments can be found in Appendix \ref{app:exp}.
\vspace{-0.2cm}
\begin{table}[ht]
\centering
\resizebox{0.8\textwidth}{!}{%
\begin{tabular}{lccccccc}
\toprule
\textbf{Method} & \multicolumn{2}{c}{\textbf{Stag Hunt}} & \multicolumn{2}{c}{\textbf{Matching Pennies}} & \multicolumn{2}{c}{\textbf{$\epsilon$-RPS}} \\
\cmidrule(lr){2-3} \cmidrule(lr){4-5} \cmidrule(lr){6-7}
 & MSE (Ref.) & Cost & MSE (Ref.) & Cost & MSE (Ref.)& Cost  \\
\midrule
SIARc   & 3.48 $\times 10^{-2}$  & 5.02 $\times 10^1$ & 5.72 $\times 10^{-2}$ & 2.26 $\times 10^2$ & 9.66 $\times 10^{-3}$  & 2.68 $\times 10^1$ \\
PINN       & 5.62 $\times 10^{-1}$   & 1.09 $\times 10^3$ & 7.70 $\times 10^{-2}$  & 2.69 $\times 10^2$ &  2.94 $\times 10^{-2}$  & 1.03 $\times 10^2$ \\
SINDYc& 6.67 $\times 10^{-1}$  & 1.27 $\times 10^3$  & 2.01 $\times 10^{-1}$ & 3.00 $\times 10^9$   & 5.77 $\times 10^{-2}$  & 7.34 $\times 10^{34}$ \\
\bottomrule
\end{tabular}%
}
\caption{Performance comparison of SIAR-MPC, PINN-MPC, and SINDY-MPC across three games described in Section \ref{sec:exp_sh}-\ref{sec:exp_zs} averaged over 100 initial conditions, evaluated on \textbf{MSE(Ref.)}: mean squared error between the estimated and reference trajectories and \textbf{Cost}: accumulated control cost. 
}
\label{tab:summary}
\end{table}
\vspace{-0.8cm}
\begin{table}[ht]
\centering
\resizebox{0.56\textwidth}{!}{%
\begin{tabular}{lccccc}
\toprule
\textbf{Method} & \multicolumn{2}{c}{\textbf{Stag Hunt}} & \multicolumn{2}{c}{\textbf{$2\times2$ Zero-sum Games}} \\
\cmidrule(lr){2-3} \cmidrule(lr){4-5}
 & MSE (Ref.) & Cost & MSE (Ref.) & Cost\\
\midrule
SIARc   & 1.95 $\times 10^{-1}$  & 3.71 $\times 10^2$ & 1.88 $\times 10^{-2}$ & 7.24 $\times 10^1$ \\
PINN       & 5.63 $\times 10^{-1}$   & 1.08 $\times 10^3$ & 3.99 $\times 10^{-2}$  & 1.31 $\times 10^2$ \\
SINDYc& 5.13 $\times 10^{-1}$  & 1.70 $\times 10^7$  & 4.88 $\times 10^{-2}$ & 3.52 $\times 10^{37}$ \\
\bottomrule
\end{tabular}%
}
\caption{Performance comparison of SIAR-MPC, PINN-MPC, and SINDY-MPC in stag hunt and $2\times 2$ zero-sum games across 50 payoff matrices, based on the metrics described in Table \ref{tab:summary}.}
\label{tab:summary2}
\end{table}
\vspace{-0.7cm}
\section{Conclusion}
\vspace{-0.1cm}
In this work we introduced \SIARMPC{}, a new computational framework for steering game dynamics towards desirable outcomes with limited data.  \SIARMPC{} extends \ac{SIAR} for system identification of controlled game dynamics and integrates it with \ac{MPC} for dynamic incentive adjustments.
Our results demonstrated that \SIARMPC{} effectively steers systems towards optimal equilibria, stabilizes chaotic and cycling dynamics.
Future research can explore several potential directions.
First, we intend to address a broader question: given the inherent limitations of game dynamics, where convergence to NE is not always guaranteed, what are the necessary and sufficient conditions for achieving global stability in controlled game dynamics?
Second, we aim to extend our framework to encompass games beyond the normal form, thereby expanding its applicability to a wider range of strategic interactions.
Finally, we plan to investigate the scalability of our approach by exploring the uncoupling assumption, which could enable its application to larger systems with numerous agents.

\acks{This research is supported by the MOE Tier 2 grant (MOE-T2EP20223-0018); the National Research Foundation, Singapore, A*STAR under the Quantum Engineering Programme (NRF2021-QEP2-02-P05); and DSO National Laboratories under its AI Singapore Program (AISG2-RP-2020-016).}

\bibliography{L4DC/references,L4DC/references_J, L4DC/references_W}

\clearpage
\appendix
\section{Additional Experimental Results}\label{app:exp}
\subsection{Performance Across Multiple Initial Conditions}
To supplement the detailed analyses presented in Section \ref{sec:exp}, we conducted simulations across a wide range of diverse settings to obtain a statistically significant understanding of each methods' performance. For each setting, we generated 100 random initial conditions, and applied the \SIARMPC{}, \PINNMPC{} and \SINDYMPC{} frameworks. The first three columns in the tables below present different metrics used to assess the performance of each framework, focusing on tracking accuracy, error at final simulation time and cost efficiency. Additionally, we generated 100 trajectories using \SIARc{}, \ac{PINN} and \SINDYc{} to evaluate how accurately these models can replicate the true system dynamics.
\subsubsection{Stag Hunt Game}
We generated 100 initial conditions using Latin hypercube sampling and then applied \SIARMPC{}, \PINNMPC{} and \SINDYMPC{} to the stag hunt game described in Section \ref{sec:exp_sh}. 
\begin{table}[ht]
\centering
\resizebox{\textwidth}{!}{%
\begin{tabular}{lcccccccc}
\toprule
\textbf{Method} & \multicolumn{2}{c}{\textbf{MSE (Ref.)}} & \multicolumn{2}{c}{\textbf{Error at $T_\text{final}$}} & {\textbf{Cost}} & \multicolumn{2}{c}{\textbf{ MSE (True)}} \\
\cmidrule(lr){2-3} \cmidrule(lr){4-5} \cmidrule(lr){7-8}
 & $x_{1,1}$ & $x_{2,1}$ & $x_{1,1}$ & $x_{2,1}$ &  & $x_{1,1}$ & $x_{2,1}$ \\
\midrule
SIARc   & 3.41 $\times 10^{-2}$  & 3.54 $\times 10^{-2}$ & 2.26 $\times 10^{-4}$ & 1.60 $\times 10^{-4}$ & 5.02 $\times 10^{1}$ & 1.21 $\times 10^{-12}$  & 3.32 $\times 10^{-9}$ \\
PINN       & 5.65 $\times 10^{-1}$   & 5.59 $\times 10^{-1}$ & 6.55 $\times 10^{-1}$  & 6.55 $\times 10^{-1}$ & 1.09 $\times 10^3$ & 7.37 $\times 10^{-3}$  & 7.89 $\times 10^{-3}$ \\
SINDYc& 6.67 $\times 10^{-1}$  & 6.67 $\times 10^{-1}$  & 7.80 $\times 10^{-1}$ & 7.80 $\times 10^{-1}$ & 1.27 $\times 10^3$  & 1.05 $\times 10^3$  & 6.50 $\times 10^2$ \\
\bottomrule
\end{tabular}%
}
\caption{Results for the stag hunt game across 100 initial conditions. The first three metrics compare the performance of the SIAR-MPC, PINN-MPC and SINDY-MPC frameworks: \textbf{MSE(Ref.)} shows the mean squared error between the estimated and the reference trajectories; \textbf{Error at $T_\text{final}$} measures the absolute error between the system state $x(T_\text{final})$ and the reference state $x^*$ at the final simulation time, $T_\text{final}$; \textbf{Cost} represents the accumulated cost of steering the system to the reference state. The last metric, \textbf{MSE(True)}, evaluates the accuracy of the SIARc, PINN, and SINDYc models in predicting the true system dynamics by reflecting the average squared error between the estimated $\dot{x}$ and the true $\dot{x}$. All results are averaged across the initial conditions. The results clearly demonstrate that SIAR-MPC and SIARc consistently achieve lower error values and control cost compared to the other methods.}
\label{tab:SH}
\end{table}
\clearpage
\subsubsection{Matching Pennies Game}
We generated 100 initial conditions using Latin hypercube sampling. We then applied \SIARMPC{}, \PINNMPC{} and \SINDYMPC{} to the matching pennies game described in Section \ref{sec:exp_zs}.  However, as the latter two methods are data-driven techniques, for fairness we did the comparison based on a larger training dataset of $K = 50$ samples.
\begin{table}[ht!]
\centering
\resizebox{\textwidth}{!}{%
\begin{tabular}{lcccccccc}
\toprule
\textbf{Method} & \multicolumn{2}{c}{\textbf{MSE (Ref.)}} & \multicolumn{2}{c}{\textbf{Error at $T_\text{final}$}} & {\textbf{Cost}} & \multicolumn{2}{c}{\textbf{MSE (True)}} \\
\cmidrule(lr){2-3} \cmidrule(lr){4-5} \cmidrule(lr){7-8}
 & $x_{1,1}$ & $x_{2,1}$ & $x_{1,1}$ & $x_{2,1}$ &  & $x_{1,1}$ & $x_{2,1}$ \\
\midrule
SIARc   & 5.13 $\times 10^{-2}$  & 6.30 $\times 10^{-2}$ & 9.90 $\times 10^{-2}$ & 1.00 $\times 10^{-1}$ & 2.26 $\times 10^2$ & 6.25 $\times 10^{-4}$ & 1.32 $\times 10^{-3}$ \\
PINN       & 7.56 $\times 10^{-2}$  & 7.83 $\times 10^{-2}$ & 2.13 $\times 10^{-1}$  & 1.88 $\times 10^{-1}$ & 2.69 $\times 10^2$ & 6.50 $\times 10^{-3}$  & 6.26 $\times 10^{-3}$ \\
SINDYc & 9.62 $\times 10^{-2}$  & 1.05 $\times 10^{-1}$  & 2.71 $\times 10^{-1}$ & 2.90 $\times 10^{-1}$ & 3.00 $\times 10^9$  & 7.97 $\times 10^2$ & 2.31 $\times 10^5$ \\
\bottomrule
\end{tabular}%
}
\caption{Results for the matching pennies game across 100 initial conditions. We compare the performance of the SIARc, PINN and SINDYc models across the metrics described in Table \ref{tab:SH}. The results clearly demonstrate that SIARc consistently achieve lower error values and control cost compared to the other methods.}
\label{tab:MP}
\end{table}
\subsubsection{$\epsilon$-RPS Game}
We generated 100 initial conditions by uniformly sampling from the 3-dimensional simplex and then applied \SIARMPC{}, \PINNMPC{} and \SINDYMPC{} to the $\epsilon$-RPS game described in Section \ref{sec:exp_zs}. 

\begin{table}[ht!]
\centering
\resizebox{\textwidth}{!}{%
\begin{tabular}{lccccccccc}
\toprule
\textbf{Method} & \multicolumn{4}{c}{\textbf{MSE (Ref.)}} & \multicolumn{4}{c}{\textbf{Error at $T_\text{final}$}} & {\textbf{Cost}} \\
\cmidrule(lr){2-5} \cmidrule(lr){6-9}
 & $x_{1,1}$ & $x_{1,2}$ & $x_{2,1}$ & $x_{2,2}$ & $x_{1,1}$ & $x_{1,2}$ & $x_{2,1}$ & $x_{2,2}$ &  \\
\midrule
SIARc   & 9.73 $\times 10^{-3}$ & 1.09 $\times 10^{-2}$ & 1.01 $\times 10^{-2}$ &7.90 $\times 10^{-3}$  &3.03 $\times 10^{-3}$  &2.14 $\times 10^{-3}$  &3.14 $\times 10^{-3}$  &4.92 $\times 10^{-3}$  & 2.68 $\times 10^1$ \\
PINN       &3.50 $\times 10^{-2}$  &3.65 $\times 10^{-2}$  &1.71 $\times 10^{-2}$  &2.89 $\times 10^{-2}$  & 1.06 $\times 10^{-1}$ & 1.44 $\times 10^{-1}$ &5.98 $\times 10^{-2}$  &7.45 $\times 10^{-2}$  & 1.03 $\times 10^2$ \\
SINDYc      &7.14 $\times 10^{-2}$ &5.29 $\times 10^{-2}$  & 5.54 $\times 10^{-2}$ & 5.11 $\times 10^{-2}$ &2.39 $\times 10^{-1}$  &2.10 $\times 10^{-1}$  &1.92 $\times 10^{-1}$  & 2.06 $\times 10^{-1}$ & 7.04 $\times 10^{34}$ \\
\bottomrule
\end{tabular}%
}
\caption{Results for the $\epsilon$-RPS game for 100 initial conditions. We compare the performance of the SIARc, PINN and SINDYc models across the metrics MSE (Ref.), Error at $T_\text{final}$, and Cost as described in Table \ref{tab:SH}. The results clearly demonstrate that SIARc consistently achieve lower error values and control cost compared to the other methods.}
\label{tab:RPS_part1}
\end{table}
\begin{table}[ht!]
\centering
\scalebox{0.9}{ 
\begin{tabular}{lcccc}
\toprule
\textbf{Method} & \multicolumn{4}{c}{\textbf{MSE (True)}} \\
\cmidrule(lr){2-5}
 & $x_{1,1}$ & $x_{1,2}$ & $x_{2,1}$ & $x_{2,2}$ \\
\midrule
SIARc   & 1.18 $\times 10^{-7}$ & 2.48 $\times 10^{-7}$ & 1.02 $\times 10^{-9}$ & 3.08 $\times 10^{-10}$ \\
PINN       & 1.31 $\times 10^{-2}$ & 1.25 $\times 10^{-2}$ & 9.05 $\times 10^{-3}$ & 1.25 $\times 10^{-2}$ \\
SINDYc      & 7.48 $\times 10^7$ & 2.92 $\times 10^4$ & 1.42 $\times 10^8$ & 1.95 $\times 10^6$ \\
\bottomrule
\end{tabular}%
}
\caption{Results for the $\epsilon$-RPS game for 100 initial conditions. We compare the performance of the SIARc, PINN and SINDYc models in terms of MSE (True) as described in Table \ref{tab:SH}. The results clearly demonstrate that SIARc consistently achieve lower error values compared to the other methods.}
\label{tab:RPS_part2}
\end{table}
\subsection{Performance Across Varying Payoff Matrices}
We conducted additional simulations to demonstrate the performance of the proposed \SIARMPC{} method when the payoff matrix is altered. 
\subsubsection{Stag Hunt Game}
We generated 50 random payoff matrices while ensuring that the game maintained its cooperative nature with a socially optimal equilibrium. We evaluated the performance of \SIARMPC{}, \PINNMPC{} and \SINDYMPC{} frameworks across these 50 systems based on our performance metrics detailed in Table \ref{tab:SH}.
\begin{table}[ht]
\centering
\resizebox{\textwidth}{!}{%
\begin{tabular}{lcccccc}
\toprule
\textbf{Method} & \multicolumn{2}{c}{\textbf{MSE (Ref.)}} & \multicolumn{2}{c}{\textbf{Error at $T_\text{final}$}} & {\textbf{Cost}} \\
\cmidrule(lr){2-3} \cmidrule(lr){4-5} 
 & $x_{1,1}$ & $x_{2,1}$ & $x_{1,1}$ & $x_{2,1}$ &   \\
\midrule
SIARc   & 2.10 $\times 10^{-1}$  & 1.80 $\times 10^{-1}$ & 1.80 $\times 10^{-1}$ & 1.80 $\times 10^{-1}$ & 3.71 $\times 10^{2}$ \\
PINN       & 5.78 $\times 10^{-1}$   & 5.47 $\times 10^{-1}$ & 5.80 $\times 10^{-1}$  & 5.80 $\times 10^{-1}$ & 1.08 $\times 10^3$ \\
SINDYc& 5.38 $\times 10^{-1}$  & 4.87 $\times 10^{-1}$  & 5.52 $\times 10^{-1}$ & 5.36 $\times 10^{-1}$ & 1.70 $\times 10^7$  \\
\bottomrule
\end{tabular}%
}
\caption{Results for the stag hunt game for 50 different payoffs. We compare the performance of the SIARc, PINN and SINDYc models across the metrics MSE (Ref.), Error at $T_\text{final}$, and Cost as described in Table \ref{tab:SH}. The results clearly demonstrate that SIARc consistently achieve lower error values and control cost compared to the other methods.}
\label{tab:SH2}
\end{table}

\subsubsection{$2 \times 2$ Zero-Sum Games}
For this experiment, 50 random payoff matrices were generated such that each zero-sum game possessed a unique mixed NE in the interior of the strategy space. The objective was to steer the system towards this interior NE. We evaluated the performance of \SIARMPC{}, \PINNMPC{} and \SINDYMPC{} frameworks across these 50 systems based on our performance metrics detailed in Table \ref{tab:SH}.
\begin{table}[ht!]
\centering
\resizebox{\textwidth}{!}{%
\begin{tabular}{lcccccc}
\toprule
\textbf{Method} & \multicolumn{2}{c}{\textbf{MSE (Ref.)}} & \multicolumn{2}{c}{\textbf{Error at $T_\text{final}$}} & {\textbf{Cost}}\\
\cmidrule(lr){2-3} \cmidrule(lr){4-5}
 & $x_{1,1}$ & $x_{2,1}$ & $x_{1,1}$ & $x_{2,1}$ & \\
\midrule
SIARc   & 1.82 $\times 10^{-2}$  & 1.94 $\times 10^{-2}$ & 2.67 $\times 10^{-2}$ & 4.41 $\times 10^{-2}$ & 7.24 $\times 10^1$  \\
PINN       & 3.88 $\times 10^{-2}$  & 4.10 $\times 10^{-2}$ & 1.26 $\times 10^{-1}$  & 1.49 $\times 10^{-1}$ & 1.31 $\times 10^2$ \\
SINDYc & 5.51 $\times 10^{-2}$  & 4.20 $\times 10^{-2}$  & 1.77 $\times 10^{-1}$ & 1.71 $\times 10^{-1}$ & 3.52 $\times 10^{37}$ \\
\bottomrule
\end{tabular}%
}
\caption{Results for 50 distinct $2 \times 2$ zero-sum games. We compare the performance of the SIARc, PINN and SINDYc models across the metrics MSE (Ref.), Error at $T_\text{final}$, and Cost as described in Table \ref{tab:SH}. The results clearly demonstrate that SIARc consistently achieve lower error values and control cost compared to the other methods.}
\label{tab:MP2}
\end{table}
\clearpage
\subsection{State Avoidance Constraints}

In this section, we revisit the example of a matching pennies game to illustrate a key feature of \ac{MPC}: its ability to incorporate state constraints \citep{BOO:CB07}.
In our case, the purpose of these state constraints is to keep the system trajectory away from specific, undesirable areas of the state space.
Consider a scenario where the policymaker aims to steer the system towards the mixed NE at $\action^{*}_{1, 1} = \action^{*}_{2, 1} = 1 / 2$, while simultaneously ensuring that the system trajectory avoids a ball of radius $0.4$ centered at $[\vect{0, 1}, \vect{1, 0}]$. 
This constraint can be incorporated into the \ac{MPC} problem and as illustrated in \cref{fig:ReplicatorDynamicsMatchingPenniesSetAvoidance}, \SIARMPC{} can steer the system to the desired equilibrium while also avoiding the restricted area of the state space (shaded in blue).
\begin{figure}[htb!]
    \centering
    \includegraphics[width=0.485\textwidth]{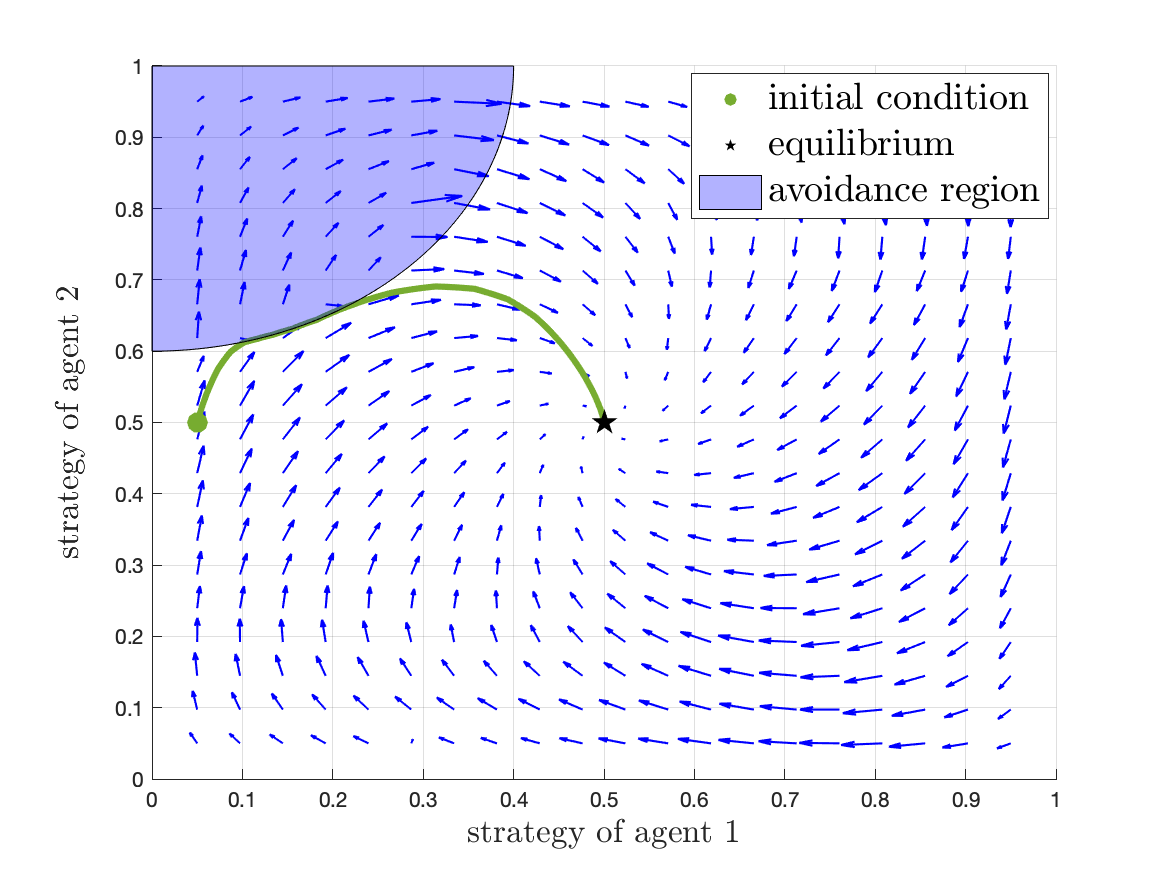}
    \includegraphics[width=0.485\textwidth]{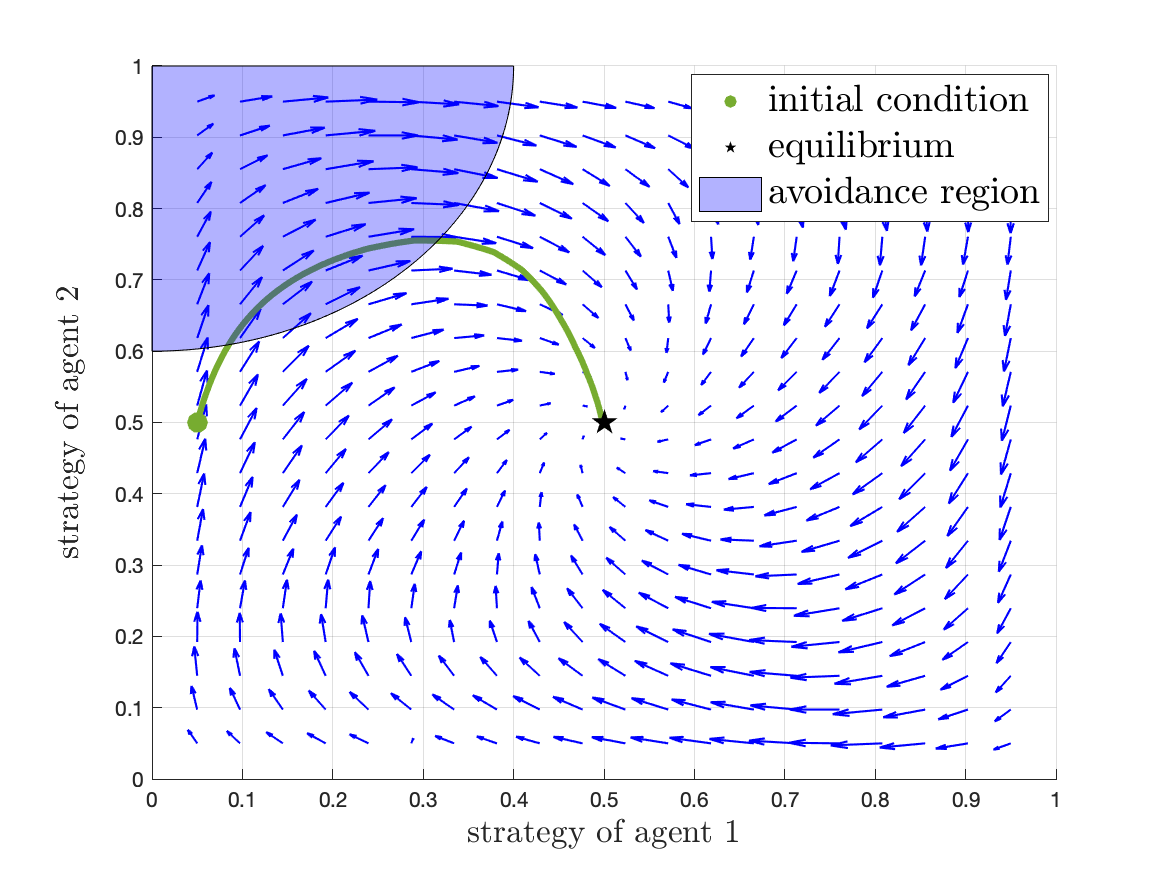}
    \caption{%
        Vector field of the steering direction of the \SIARMPC{} when the additional state avoidance constraint is~\figl{}, and is not~\figr{} imposed.
        The restricted area is shaded in blue.
        The green line corresponds to a controlled trajectory of the replicator dynamics for the matching pennies~game.
    }
\label{fig:ReplicatorDynamicsMatchingPenniesSetAvoidance}
\end{figure}
\clearpage
\section{Details of \ac{PINN}}\label{app:pinn}
Consider a two-player two-action game. Given that each player has two actions, the state space can be reduced to $(x_1,x_2)\in[0,1]\times[0,1]$, where $x_1$ and $x_2$ represent the first state of each player\footnote{Note that the second state of each player is simply $1-x_i$ for $i={1,2}$, and thus, the corresponding dynamics are given by $-f_i$ for each $i={1,2}$.}. We train a neural network to approximate the dynamics $f_1$ and $f_2$ such that
$$\dot{x}_i=f_i(x_1,x_2,\omega)\approx p_i(x_1,x_2,\omega)=\dot{\hat{x}}_i,\ i={1,2},$$ where $\omega\in\Omega.$ The neural network takes $x_1,x_2,\omega$ as inputs and outputs an approximation $\dot{\hat{x}}_i$ of $\dot{x}_i$. The training dataset is generated as described in Section \ref{sec:siarstep} and is represented as $D:=\{(x_{1,i},x_{2,i},\omega_{i},\dot{x}_{1,i},\dot{x}_{2,i})\}^{N_D}_{i=1}$. 
\subsection{Physics Information and Collocation Dataset Generation}
To incorporate the physical knowledge into the \ac{PINN} training, we consider two sets of side information constraints: Robust Forward Invariance (\ac{RFI}) and Positive Correlation (\ac{PC}).\\ 

\noindent\textit{Robust Forward Invariance}\\
For the reduced state space, the \ac{RFI} constraint can be written as:
\begin{align*}
&p_1(0,x_2,\omega)\geq0,\quad \forall x_2\in[0,1],\ \omega\in\Omega\\
&p_1(1,x_2,\omega)\leq0,\quad \forall x_2\in[0,1],\ \omega\in\Omega\\
&p_2(x_1,0,\omega)\geq0,\quad \forall x_1\in[0,1],\ \omega\in\Omega\\
&p_2(x_1,1,\omega)\leq0,\quad \forall x_1\in[0,1],\ \omega\in\Omega
\end{align*}
To enforce this side information constraint during training, we generate a set of $N_C^\text{RFI}$ collocation points. Specifically, we create the following sets:
\begin{align*}
C^\text{RFI}_1:=\{(x_{1,i},x_{2,i},\omega_i)\}_{i=1}^{N_{C_1}^{\text{RFI}}}\ \text{with}\ x_{1,i}=0,\ x_{2,i}\sim U[0,1]\ \text{and}\ \omega_i\sim U\Omega\\
C^\text{RFI}_2:=\{(x_{1,i},x_{2,i},\omega_i)\}_{i=1}^{N_{C_2}^{\text{RFI}}}\ \text{with}\ x_{1,i}=1,\ x_{2,i}\sim U[0,1]\ \text{and}\ \omega_i\sim U\Omega\\
C^\text{RFI}_3:=\{(x_{1,i},x_{2,i},\omega_i)\}_{i=1}^{N_{C_3}^{\text{RFI}}}\ \text{with}\ x_{2,i}=0,\ x_{1,i}\sim U[0,1]\ \text{and}\ \omega_i\sim U\Omega\\
C^\text{RFI}_4:=\{(x_{1,i},x_{2,i},\omega_i)\}_{i=1}^{N_{C_4}^{\text{RFI}}}\ \text{with}\ x_{2,i}=1,\ x_{1,i}\sim U[0,1]\ \text{and}\ \omega_i\sim U\Omega
\end{align*}
where $x\sim U\mathcal{S}$ denotes that $x$ is sampled uniformly from set $\mathcal{S}$.
The combined set of collocation points for \ac{RFI} is then $C^\text{RFI}=C^\text{RFI}_1\cup C^\text{RFI}_2\cup C^\text{RFI}_3\cup C^\text{RFI}_4.$\\

\noindent\textit{Positive Correlation}\\
For the reduced state space, the \ac{PC} constraints can be written as:
\begin{align*}
&\inner{\nabla_{x_1} \payoff_1(x_1,x_2, \control)}{p_1(x_1,x_2, \control)}\geq0, \forall  x_1,x_2\in[0,1],\ \omega\in\Omega\\
&\inner{\nabla_{x_2} \payoff_2(x_1,x_2,, \control)}{p_2(x_1,x_2,, \control)}\geq0, \forall  x_1,x_2\in[0,1],\ \omega\in\Omega
\end{align*}
To enforce this side information constraint during training, we generate a set of $N_C^\text{PC}$ collocation points. Specifically, we create the following set:
$$C^\text{PC}:=\{(x_{1,i},x_{2,i},\omega_i)\}_{i=1}^{N_{C}^{\text{PC}}}\ \text{with}\ x_{1,i},\ x_{2,i}\sim U[0,1]\ \text{and}\ \omega_i\sim U\Omega.$$

\subsection{\ac{PINN} Loss Function}
The loss function for PINN integrates three components: a loss function for the supervised learning over dataset $D$ and two physics-informed loss functions over the collocation points $C^\text{RFI}$ and $C^\text{PC}$ to enforce side information constraints \ac{RFI} and \ac{PC}.\\

\noindent\textit{Supervised Loss}\\
The supervised loss is the standard squared error loss to ensure neural netwrok predictions align closely with the training data and is given as: 
\begin{align*}
\ell^0 &= \sum_{(x_1,x_2,\omega,\dot{x}_1,\dot{x}_2)\in D} (p_1(x_1,x_2, \omega) - \dot{x}_1)^2+(p_2(x_1,x_2, \omega) - \dot{x}_2)^2
\end{align*}\\

\noindent\textit{Physics Loss}\\
The physics loss for \ac{RFI} is defined as:
\begin{align*}
\ell^{\text{RFI}} = \sum_{(x_1,x_2,\omega) \in C^{\text{RFI}}_1} &\max(0, -p_1(x_1,x_2, \omega))+ \sum_{(x_1,x_2,\omega) \in C^{\text{RFI}}_2} \max(0, p_1(x_1,x_2, \omega))\\& + \sum_{(x_1,x_2,\omega) \in C^{\text{RFI}}_3} \max(0, -p_2(x_1,x_2, \omega))+ \sum_{(x_1,x_2,\omega) \in C^{\text{RFI}}_4} \max(0, p_2(x_1,x_2, \omega)).
\end{align*}
Similarly, the physics loss for the \ac{PC} constraint is defined as:
\begin{align*}
\ell^{\text{PC}} = &\sum_{(x_1,x_2, \omega) \in C^{\text{PC}}} \max\big(0, -\nabla_{x_1} \payoff_1(x_1,x_2, \control) \cdot p_1(x_1,x_2, \omega)\big) \\
&+ \sum_{(x_1,x_2, \omega) \in C^{\text{PC}}} \max\big(0, -\nabla_{x_2} \payoff_2(x_1,x_2, \control)\cdot p_2(x_1,x_2, \omega)\big).
\end{align*}
Then, the loss function for PINN with a weighted summation is given as:
\begin{align*}
\ell=\lambda_0\ell^0+\lambda_{\text{RFI}}\ell^\text{RFI}+\lambda_{\text{PC}}\ell^\text{PC}.
\end{align*}

\section{SIARc Theory}
\label{app:SIARc}

In this section, we state a formal version of Theorem \ref{thm:SIARc_informal} in the main text and provide its proof. We note that the proof is a straightforward extension of the proof given in \cite{ahmadi_learning_2023} to the setting of control ODEs; nevertheless, we write it here in full for completeness and also for cleaner exposition of certain aspects of the side information constraints and SOS certification that are specific to our setting.

The proof of the theorem depends on first using the Stone-Weierstrass theorem to find polynomial dynamics that is pointwise close to the given (true) dynamics, then using this fact to bound both (i) its distance to the true dynamics in terms of what we are interested in and (ii) how well the polynomial dynamics satisfies the side information. For organizational clarity we split the proof of the theorem into three subsections: in subsections \ref{appsec:SIARc_dist} and \ref{appsec:SIARc_sat} we deal with bounds (i) and (ii) respectively, while in subsection \ref{appsec:SIARc_thm} we state the theorem formally and complete its proof.

\subsection{Bounding the Distance between Dynamics}
\label{appsec:SIARc_dist}

Let $\X \subset \R^n$ be a compact state space, and $\Omega \subset \R^m$ be a compact space of viable controls. Let $\contxw$ be the space of continuously differentiable functions $f: \X \times \Omega \to \R^n$. Suppose that the dynamics of the true system are governed by a vector field $f \in \contxw$, i.e.,
\[
	\dot{x}(t) = f(x(t), \omega(t)),
\]
and let $x_{f,\wseq}(t, \xinit)$ denote the trajectory starting from $\xinit \in \X$ and following the dynamics of $f$, with control sequence $\wseq: [0, t] \to \Omega$.

We have the following notion of distance between any two vector fields $f,g \in \contxw$:
\begin{equation}
\label{dist}
	\dist(f,g) :=
	\sup_{(t, \xinit, \wseq) \in \txwadmis} 
	\max\{
		\norm{x_{f,\wseq}(t, \xinit) - x_{g,\wseq}(t, \xinit)}_2, \,
		\norm{\dot{x}_{f,\wseq}(t, \xinit) - \dot{x}_{g,\wseq}(t, \xinit)}_2
	\}
\end{equation}
where
\begin{equation}
\label{traj_space}
	\txwadmis 
	:=
	\{
		(t, \xinit, \wseq) \in [0, T] \times \X \times \wseqadmis:
		\xfw(s, \xinit), \xgw(s, \xinit) \in \X \fa s \in [0, t]
	\}
\end{equation}
and $\wseqadmis$ is the set of functions from $[0, T] \to \Omega$ that are continuous except at a finite number of discontinuities (the set of admissible control sequences).\footnote{$\wseqadmis$ can be chosen more generally so that all the integrals in the proof of Theorem \ref{thm:vecfielddist_bound} are integrable. That $\wseq$ is continuous except at a finite number of discontinuities is sufficient for Riemann integrability due to continuity of $f,g$ in $\omega$.}

As a ready extension of the Stone-Weierstrass theorem to vector-valued functions, we have the following lemma:

\begin{lemma}[see e.g., \cite{stone1948generalized}]
	For any compact sets $\X \subset \R^n$ and $\Omega \subset \R^m$, scalar $\epsilon > 0$, and continuous function $f: \X \times \Omega \to \R^n$, there exists a vector polynomial function $p : \R^n \times \R^m \to \R^n$ (i.e., $p(x, \omega) = (p_1(x, \omega), \ldots, p_n(x, \omega))$ where each $p_i$ is a polynomial in $x, \omega$) such that
	\[
		\norm{f-p}_{\X \times \Omega}
		\equiv
		\max_{\substack{x \in \X \\ \omega \in \Omega}}
		\norm{f(x, \omega) - p(x, \omega)}_2
		\leq
		\epsilon.
	\]
\end{lemma}

We also recall the Grönwell-Bellman inequality:
\begin{lemma}[Grönwell-Bellman Inequality (\cite{bellman1943stability})]
\label{gronwell}
	Let $I = [a,b]$ be a nonempty interval in $\R$. Let $u, \alpha, \beta: I \to \R$ be continuous functions satisfying
	\[
		u(t) \leq \alpha(t) + \int_a^t \beta(s) u(s) \ds \qfa t \in I.
	\]
	If $\alpha$ is nondecreasing and $\beta$ is nonnegative, then
	\[
		u(t) \leq \alpha(t) e^{\int_a^t \beta(s) \ds} \qfa t \in I.
	\]
\end{lemma}

We are now ready to prove the following proposition:

\begin{proposition}
\label{thm:vecfielddist_bound}
	Suppose that $\X \subset \R^n$ and $\Omega \subset \R^m$ are compact sets, $T > 0$ is a finite time horizon, and $\wseq: [0, T] \to \Omega$ is a fixed control sequence that is continuous except at a finite number of discontinuities.   Then for all $f,g \in \contxw$,
	\[
		\normxw{f-g}
		\leq
		\dist(f,g)
		\leq
		\max\{Te^{LT}, 1+LTe^{LT}\} \normxw{f-g}
	\]
	where $L$ is any scalar for which either $f$ or $g$ is $L$-Lipschitz over $\X \times \Omega$.\footnote{Since $f,g$ are continuously differentiable functions over the compact domain $\X \times \Omega$, a finite $L$ exists. Indeed, we don't actually need $f$ to be $L$-Lipschitz over $\X \times \Omega$: the weaker condition that for all $\omega \in \Omega$ the function $f(\farg, \omega): \X \to \R^n$ is $L$-Lipschitz over $\X$ suffices.}
\end{proposition}

\begin{proof}
	The first inequality holds because
	\begin{align*}
		\dot{x}_{f,\wseq}(0, \xinit) &= f(\xinit, \wseq(0)), \\
		\dot{x}_{g,\wseq}(0, \xinit) &= g(\xinit, \wseq(0))
	\end{align*}
	so
	\[
		\normxw{f-g} 
		= \max_{\substack{\xinit \in \X \\ \wseq(0) \in \Omega}} \norm{f(\xinit, \wseq(0)) - g(\xinit, \wseq(0))}_2
		\leq \dist(f,g).
	\]
	For the second inequality, without loss of generality suppose that $f$ is $L$-Lipschitz. Fix a $(\txw) \in \txwadmis$. We first bound
	\[
		\norm{
			x_{f,\wseq}(t, \xinit) - x_{g,\wseq}(t, \xinit)
		}_2.
	\]
	We have that
	\begin{align*}
		\xfw(\tx) - \xgw(\tx)
		= \
		&\int_0^t f(\xfw(\sx), \wseq(s)) - g(\xgw(\sx), \wseq(s)) \ds \\
		= \
		&\int_0^t f(\xfw(\sxw), \wseq(s)) - f(\xgw(\sx), \wseq(s)) \ds \\
		&+
		\int_0^t f(\xgw(\sxw), \wseq(s)) - g(\xgw(\sx), \wseq(s)) \ds.
	\end{align*}
	Thus, by triangle inequality
	\begin{equation}
	\begin{split}
	\label{normbound_triangle}
		\norm{ \xfw(\tx) - \xgw(\tx) }_2
		\leq \
		&\int_0^t \norm{  f(\xfw(\sx), \wseq(s)) - f(\xgw(\sx), \wseq(s)) }_2 \ds \\
		&+ \int_0^t \norm{ f(\xgw(\sx), \wseq(s)) - g(\xgw(\sx), \wseq(s)) }_2 \ds.
	\end{split}
	\end{equation}
	For any $s \in [0, t]$, we can upper bound
	\[
		\norm{  f(\xfw(\sx), \wseq(s)) - f(\xgw(\sx), \wseq(s)) }_2
		\leq
		L \norm{\xfw(\sx) - \xgw(\sx)}_2
	\]
	by the Lipschitz property of $f$, and upper bound
	\[
	 	\norm{ f(\xgw(\sx), \wseq(s)) - g(\xgw(\sx), \wseq(s)) }_2
		\leq
		\normxw{f-g}.
	\]
	Inserting these two bounds back into \eqref{normbound_triangle} gives
	\[
		\norm{ \xfw(\tx) - \xgw(\tx) }_2
		\leq
		L \int_0^t \norm{\xfw(\sx) - \xgw(\sx)}_2 \ds
		+
		t \normxw{f-g}.
	\]
	Then applying the Grönwell-Bellman inequality (Lemma \ref{gronwell}) with
	\begin{align*}
		u(t) &= \norm{ \xfw(\tx) - \xgw(\tx) }_2, \\
		\alpha(t) &= t \normxw{f-g}, \\
		\beta(t) &= L
	\end{align*}
	gives us the bound
	\begin{equation}
	\label{bound:xdist}
		\norm{ \xfw(\tx) - \xgw(\tx) }_2
		\leq
		te^{Lt} \normxw{f-g}.
	\end{equation}
	An upper bound for $\norm{ \dotxfw(\tx) - \dotxgw(\tx) }_2$ follows readily since
	\begin{equation}
	\begin{split}
	\label{bound:xdotdist}
		\norm{ \dotxfw(\tx) - \dotxgw(\tx) }_2
		= \
		&\norm{ f(\xfw(\tx), \wseq(s)) - g(\xgw(\tx), \wseq(s)) }_2 \\
		\leq \
		&\norm{ f(\xfw(\tx), \wseq(s)) - f(\xgw(\tx), \wseq(s)) }_2  \\
		&+\norm{ f(\xgw(\tx), \wseq(s)) - g(\xgw(\tx), \wseq(s)) }_2 \\
		\leq \
		&L \norm{ \xfw(\tx) - \xgw(\tx) }_2 + \normxw{f-g} \\
		\leq \
		&(1 + Lte^{Lt}) \normxw{f-g},
	\end{split}
	\end{equation}
	where the last inequality is due to \eqref{bound:xdist}. Finally, putting \eqref{bound:xdist} and \eqref{bound:xdotdist} together and using the fact that $t \in [0, T]$ gives us the uniform bound
	\begin{align*}
		&\max\{
			\norm{\xfw(t, \xinit) - \xgw(t, \xinit)}_2, \,
			\norm{\dotxfw(t, \xinit) - \dotxgw(t, \xinit)}_2
		\} \\
		&\leq
		\max\{te^{Lt}, 1 + Lte^{Lt}\} \normxw{f-g} \\
		&\leq
		\max\{Te^{LT}, 1 + LTe^{LT}\} \normxw{f-g}, \\
	\end{align*}
	and taking supremum over $(\txw) \in \txwadmis$ gives the desired bound
	\[
		\dist(f,g) \leq \max\{Te^{LT}, 1 + LTe^{LT}\} \normxw{f-g}. 
	\]
\end{proof}

\subsection{Side-Information Constraints and $\delta$-Satisfiability}
\label{appsec:SIARc_sat}

%
%

Throughout this work, we consider side information constraints $S$ for which there exists an appropriate polynomial $q: \R^n \times \R^m \to \R^n$ such that for any given dynamics $f$,
\begin{equation}
\label{satS}
	\text{$f$ satisfies $S$}
	\Longleftrightarrow
	\inp{q_S(\action, \control)}{f(\action, \control)}
	 \geq 0 \qqfa \action \in \X, \, \control \in \Omega.
\end{equation}
(Both \ref{eq:RobustForwardInvariance} and the relaxed version of \ref{eq:PositiveCorrelation} we enforce can be broken down into multiple constraints of this form.)

By defining
\[
	L_{S}(f) := \max\{0, \max_{\substack{\action \in \X, \\ \control \in \Omega}} - \inp{q_S(\action, \control)}{f(\action, \control)}\},
\]
we immediately have that $L_{S}(f) \geq 0 \fa f$ and
\begin{equation}
\begin{split}
\label{eq:LS}
	L_{S}(f) = 0 &\Longleftrightarrow \inp{q_S(\action, \control)}{f(\action, \control)} \geq 0 \ \ \qqfa \action \in \X, \, \control \in \Omega, \\
	L_{S}(f) \leq \delta &\Longleftrightarrow \inp{q_S(\action, \control)}{f(\action, \control)} + \delta \geq 0 \qfa \action \in \X, \, \control \in \Omega.
\end{split}
\end{equation}

$\ls(f)$ can thus be taken as a measure of how well the dynamics $f$ satisfies the side information constraint $S$, and we say that dynamics $f$ \emph{$\delta$-satisfies} $S$ if $L_{S}(f) \leq \delta$.

%

We also have the following lemma ensuring that if two dynamics are bounded in the $\normxw{\farg}$ norm then the difference between their $\ls$ values are bounded. In particular, this implies that if $f$ exactly satisfies $S$ (i.e., $\ls(f) = 0$) and $f, g$ are sufficiently close, then $g$ $\delta$-satisfies $S$.

\begin{proposition}
\label{lem:LSnormbound}
For side constraints $S$ of the form \eqref{satS} we have that 
\[
	\forall \delta > 0 \ \exists \, \epsilon > 0: \quad \normxw{f-g} \leq \epsilon \ \Longrightarrow \ \abs{L_S(f) - L_S(g)} \leq \delta.
\]
\end{proposition}

\begin{proof}
	Note that we need only concern ourselves with $f,g$ for which $\ls(f) \neq \ls(g)$. We can then without loss of generality assume that $\ls(f) > \ls(g)$. We shall use in our proofs the fact that if $h_1, h_2 : \mathcal{Y} \to \R$, then
			\begin{equation}
			\label{eq:hfact}
				\max_y h_1(y) - \max_y h_2(y)
				\leq 
				\max_y \abs{h_1(y) - h_2(y)}.
			\end{equation}
			
	We then have that
			\begin{align*}
				&\abs{\ls(f) - \ls(g)}  &\\
				&\qquad= 
					\ls(f) - \ls(g) \qquad&\\
				&\qquad= 
					\max_{\action, \control}( - \inp{q_S(\action, \control)}{f(\action, \control)}) 
					- \ls(g)
					\qquad
					&\text{since } \ls(f) > 0 \\
				&\qquad
					\leq \max_{\action, \control}( - \inp{q_S(\action, \control)}{f(\action, \control)})
					- \max_{\action, \control, i}( - \inp{q_S(\action, \control)}{g(\action, \control)})
					\qquad
					&\text{since } \ls(g) \geq 0 \\
				&\qquad\leq
					\max_{\action, \control}
						\abs{
							- \inp{q_S(\action, \control)}{f(\action, \control)}
							+  \inp{q_S(\action, \control)}{g(\action, \control)}
							}
					\qquad
					&\text{by } \eqref{eq:hfact}\\
				&\qquad=
					 \max_{\action, \control}
					 	\abs{ \inp{ q_S(\action, \control) }{ f(\action, \control) - g(\action, \control) } }
					& \\
				&\qquad\leq
					\max_{\action, \control}
						\left(
							\norm{ q_S(\action, \control) }_2
							\norm{ f(\action, \control) - g(\action, \control) }_2 
						\right)
					&\text{by Cauchy-Schwarz} \\
				&\qquad\leq
						\max_{\action, \control} (\norm{ q_S(\action, \control) }_2)
						\cdot
						\max_{\action, \control} (\norm{ f(\action, \control) - g(\action, \control) }_2)
					& \\
				&\qquad=
					 \normxw{q_S} \cdot \normxw{f-g}  .
			\end{align*}
			Finally, $\normxw{q_S} < \infty$ since the polynomial $q_S$ is continuous over $\X \times \Omega$. We can thus set
			\[
				\epsilon := \frac{\delta}{ \normxw{q_S}}
			\]
			to ensure that
			\[
				\normxw{f-g} \leq \epsilon
				\Longrightarrow
				\abs{\ls(f) - \ls(g)} < \delta. 
			\]
\end{proof}

\subsection{Main Theorem (Formal) and Proof}
\label{appsec:SIARc_thm}

We are now in a position to state and prove the formal theorem (an informal version of which was stated in Theorem \ref{thm:SIARc_informal} of the main text):
\begin{customthm}{1}[Main Theorem, Formal]
\label{thm:SIAR}
	Fix a time horizon $T$, desired approximation accuracy $\epsilon > 0$, and desired side information accuracy $\delta > 0$. Consider any continuously differentiable dynamics $f \in \contxw$ that satisfies side information constraints $S_k$ for $k = 1, 2, \ldots, K$ that can be expressed in the form \eqref{satS}. Then there exists polynomial dynamics $p: \R^n \times \R^m \to \R^n$ that is $\epsilon$-close to $f$ and $\delta$-satisfies the same side information, i.e., satisfies
	\[
		\dist(f,p) \leq \epsilon \qquad \text{and} \qquad L_{S_k}(p) \leq \delta, \ k = 1, \ldots, K.
	\]
	Furthermore $L_{S_k}(p) \leq \delta$ has an SOS certificate, i.e.,
	\[
		\inp{q_{S_k}(\action, \control)}{p(\action, \control)} + \delta \in QM(\X \times \Omega),
	\]
	where $QM$ is the quadratic module of (an appropriate description) of $\X \times \Omega$.
	
%
\end{customthm}

\begin{proof}
	In the first step, we rely on the Stone-Weierstrass theorem to show the existence of a polynomial $p$ that satisfies $\normxw{p-f} \leq \gamma$ for an appropriately chosen $\gamma$ (which depends on $\epsilon, \delta$)
	and hence, by Propositions \ref{thm:vecfielddist_bound} and \ref{lem:LSnormbound} respectively, is $\epsilon$-close to $f$ and $\frac{\delta}{2}$-satisfies the side information, i.e., 
	\[
		\dist(f,p) \leq \epsilon \qquad \text{and} \qquad L_{S_k}(p) \leq \frac{\delta}{2}, \ k = 1, \ldots, K.
	\]
	Consequently, since by \eqref{eq:LS} we have that $\inp{q_{S_k}(\action, \control)}{p(\action, \control)} + \frac{\delta}{2} \geq 0 \fa x \in \X$, so adding $\frac{\delta}{2}$ gives
	\[
		\inp{q_{S_k}(\action, \control)}{p(\action, \control)} + \delta >  0 \qfa \action \in \X, \, \control \in \Omega.
	\]
	Finally, by Putinar's Positivstellensatz (\cite{putinar1993positive}), since an appropriate description\footnote{
		That a semialgebraic set satisfies the Archimedean property is description-dependent. Since $\X \times \Omega$ is compact with an easily obtainable bound $R$ on its Euclidean norm, the redundant constraint $R - \norm{(\action, \control)}_2^2 \geq 0$ can be added to the description of the feasible set $\X \times \Omega$ to make it Archimedean so that we may apply Putinar's Positivstellensatz. (See, e.g., the discussion in \cite{laurent_sums_2008} Section 3.6 for more details.)
	} of $\X \times \Omega$ is Archimedean, we have that $\inp{q_{S_k}(\action, \control)}{p(\action, \control)} + \delta \in QM(\X \times \Omega)$, i.e., there exists a sum-of-squares certificate that certifies that $L_{S_k}(p) \leq \delta$ (i.e., that $q_S(x,p(x)) + \delta \geq 0$ on $\X \times \Omega$).
\end{proof}

What this theorem means for SIARc is that, given true dynamics dynamics $f$ we are guaranteed to have polynomial $p$ that is close to $f$ (in the $\dist$ distance) and $\delta$-satisfies the side information; furthermore, since $\delta$-satisfiability has an SOS certificate we can search over this space in our regression problem by using the SOS hierarchy.

\end{document}